\documentclass[11pt]{article}

\usepackage{amsfonts}
\usepackage{amsmath}
\usepackage{amssymb}
\usepackage{amsthm}
\usepackage{mathrsfs}
\usepackage{psfrag}
\usepackage{color}

\usepackage{graphicx}

\theoremstyle{plain}
\newtheorem{theorem}{Theorem}
\newtheorem{proposition}[theorem]{Proposition}
\newtheorem{definition}[theorem]{Definition}
\newtheorem{lemma}[theorem]{Lemma}
\newtheorem{corollary}[theorem]{Corollary}

\theoremstyle{definition}
\newtheorem{example}[theorem]{Example}
\newtheorem{remark}[theorem]{Remark}

\newtheorem{conjecture}[theorem]{Conjecture}

\numberwithin{exercise}{section} \numberwithin{equation}{section}
\numberwithin{theorem}{section} \numberwithin{problem}{section}

\numberwithin{figure}{section}

 \DeclareMathOperator{\diag}{diag}

\DeclareMathOperator{\diam}{diam}

\newcommand{\bs}[1]{{\boldsymbol{#1}}}

\newcommand{\R}{\mathbf{R}}
\newcommand{\N}{\mathbf{N}}
\newcommand{\Z}{\mathbf{Z}}
\newcommand{\Q}{\mathbf{Q}}
\newcommand{\C}{\mathbf{C}}
\newcommand{\I}{\mathrm{i}}

\begin{document}
\title{Generalized Quasispecies Model on Finite Metric Spaces: Isometry Groups and Spectral Properties of Evolutionary Matrices}

\author{Yuri S. Semenov$^{{1},}$\footnote{yuri\_semenoff@mail.ru}\,\,, Artem S. Novozhilov$^{{2},}$\footnote{artem.novozhilov@ndsu.edu} \\[3mm]
\textit{\normalsize $^\textrm{\emph{1}}$Applied Mathematics--1, Moscow State University of Railway Engineering,}\\[-1mm]\textit{\normalsize Moscow 127994, Russia}\\[2mm]
\textit{\normalsize $^\textrm{\emph{2}}$Department of Mathematics,
North Dakota State University, Fargo, ND 58108, USA}}

\date{}

\maketitle

\begin{abstract}
The quasispecies model introduced by Eigen in 1971 has close
connections with the isometry group of the space of binary
sequences relative to the Hamming distance metric. Generalizing
this observation we introduce an abstract quasispecies model on a
finite metric space $X$ together with a group of isometries
$\Gamma$ acting transitively on $X$. We show that if the domain of
the fitness function has a natural decomposition into the union of
$t$ $G$-orbits, $G$ being a subgroup of $\Gamma$, then the
dominant eigenvalue of the evolutionary matrix satisfies an
algebraic equation of degree at most $t\cdot {\rm rk}_{\Z} R$,
where $R$ is what we call the orbital ring. The general theory is
illustrated by two examples, in both of which $X$ is taken to be the metric
space of vertices of a regular polytope with the ``edge'' metric;
namely, the case of a regular $m$-gon and of a hyperoctahedron are
considered.

\end{abstract}

\section{Introduction}
The \textit{quasispecies model}, initially put forward by Manfred Eigen in \cite{eigen1971sma} to comprehensively study the problem of the origin of life, is now a classical object of modern evolutionary theory. More pertinent for the present paper, this model possesses a rich internal mathematical structure, as first was noted in \cite{dress1988evolution,rumschitzki1987spectral}, where intriguing connections between evolutionary dynamics on sequence space and tensor products of representation spaces were pointed out. This mathematical framework, interesting on its own, facilitates understanding why some versions of Eigen's model can be solved exactly and why for some other innocently looking versions numerical computations and subtle approximations are required. In \cite{semenov2016eigen} we noticed and used similar connections to introduce and analyze a special case of Eigen's model, in which two different types of sequences are present; we also formulated, using geometric language, an abstract mathematical model, which we called the \textit{generalized quasispecies} or \textit{Eigen model}. The goal of this paper is to present in detail, expand, and elaborate on this generalized model with the ultimate objective to outline a proper mathematical framework in which many peculiarities of the Eigen model, including the notorious \textit{error threshold},  can be understood from an algebraic point of view.

Eigen's model is quite special in bringing together abstract mathematics and biology. Even more uniquely, it also has very tight connections with statistical mechanics. The complexity and richness of the original Eigen's model can be emphasized by the fact that it is equivalent to the famous \textit{Ising model} in statistical mechanics \cite{leuthausser1987statistical,leuthausser1986exact}. The Ising model can be solved exactly only in some special cases, and hence any progress in understanding the conditions to solve Eigen's model may yield insights in the analysis of the Ising model.

In what follows we neither aim for the most general formulation of
the quasispecies model keeping the mutations symmetric and
independent, nor we present the most abstract version of our
model, using as the specific examples of the underlying metric
spaces regular polytopes with natural ``edge'' metrics. In this
way the presentation, in our opinion, can be accessible to
theoretical biologists, physicists, and mathematicians alike. The
rest of the text is organized as follows. In Section \ref{sec:2}
we recall the classical Eigen's model, provide a concise
description of the main mathematical advances of its analysis and
show in which way the hyperoctahedral group of isometries of the
space of two-letter sequences with the Hamming distance naturally
appears in the analysis of this model. This sets the stage for an
abstract formulation of the generalized Eigen's model on an
arbitrary finite metric space in Section \ref{sec:3}. In the same
section we also review the necessary algebraic background and
introduce what we call \textit{an orbital ring} that allows
identifying those spectral problems for which progress can be
achieved. Section \ref{sec:4} contains an explicit equation for
the dominant eigenvalue. In Section \ref{sec:5} we apply the
abstract theory developed so far to two specific cases, namely, to
the regular $m$-gon and to the hyperoctahedral mutational
landscapes. Short Section 6 is devoted to the discussion of open problems and future directions. Finally, Appendix contains some additional
calculations in a concise table form.

\section{The quasispecies model}\label{sec:2}
The quasispecies model \cite{eigen1971sma,eigen1988mqs} is a system of ordinary differential equations that describes the changes with time of the vector of frequencies of different types of individuals in a population. To be specific, the individuals are defined to be sequences of a fixed length, say $N$, composed of a two-letter alphabet $\{0,1\}$, hence we have $2^N=:l$ different types of sequences. Sequences can reproduce and mutate; the former is incorporated into the diagonal matrix $\bs W=\diag (w_0,\ldots,w_{l-1})$, which is called the \textit{fitness landscape}, and the latter is described by the stochastic matrix $\bs Q$, which is called the \textit{mutation landscape}. The entry $w_i\geq 0$ of $\bs W$ is the fitness of the sequence of type $i$, the entry $q_{ij}\in[0,1]$ of $\bs Q$ is interpreted as the probability that, upon reproduction, the sequence of type $j$ begets the sequence of type $i$. It is readily shown that the asymptotic state of the vector $\bs{\hat p}=(\hat p_0,\ldots,\hat p_{l-1})^\top\in \R^{l}$ of frequencies of different types of sequences is the positive eigenvector corresponding to the dominant eigenvalue $\overline{w}$ of the eigenvalue problem
\begin{equation}\label{eq0:1}
    \bs{QW\hat p}=\overline{w}\bs{\hat p}.
\end{equation}
The dominant eigenvalue and the corresponding eigenvector exist under some very mild technical conditions on $\bs{W}$ and $\bs{Q}$ due to the Perron--Frobenius theorem. The leading eigenvalue $\overline{w}$ is called the \textit{mean population fitness} and is given by $\overline{w}=\sum_{i=0}^{l-1}w_i\hat p_i$.  (We note that there exists an equally popular evolutionary model, which is usually called the Crow--Kimura model, whose properties are close to the problem \eqref{eq0:1}, see, e.g., \cite{baake1999,bratus2013linear,semenov2015}. Much more on the history and analysis of the various quasispecies models can be found in \cite{baake1999,schuster2015quasispecies,jainkrug2007}.)

To make further progress one needs to specify matrices $\bs W$ and $\bs Q$. In the simplest symmetric case we can assume that mutation at a given site of a sequence is independent from other mutations, and the mutation probability, which we denote $1-q$, such that $q$ is the \textit{fidelity}, i.e., the probability of the error free reproduction, is the same for any site. Then
$$
q_{ij}=q^{N-H_{ij}}(1-q)^{H_{ij}},\quad i,j=0,\ldots,l-1,
$$
where $H_{ij}$ is the Hamming distance between sequences of types $i$ and $j$ (we use the lexicographical order to index the sequences, such that sequence $i$ is given by the binary representation of length $N$ of the integer $i$). Thus the model has the natural geometry of the binary hypercube $X=\{0,1\}^N$, see Fig. \ref{fig:1}.
\begin{figure}[!ht]
\centering
\includegraphics[width=0.7\textwidth]{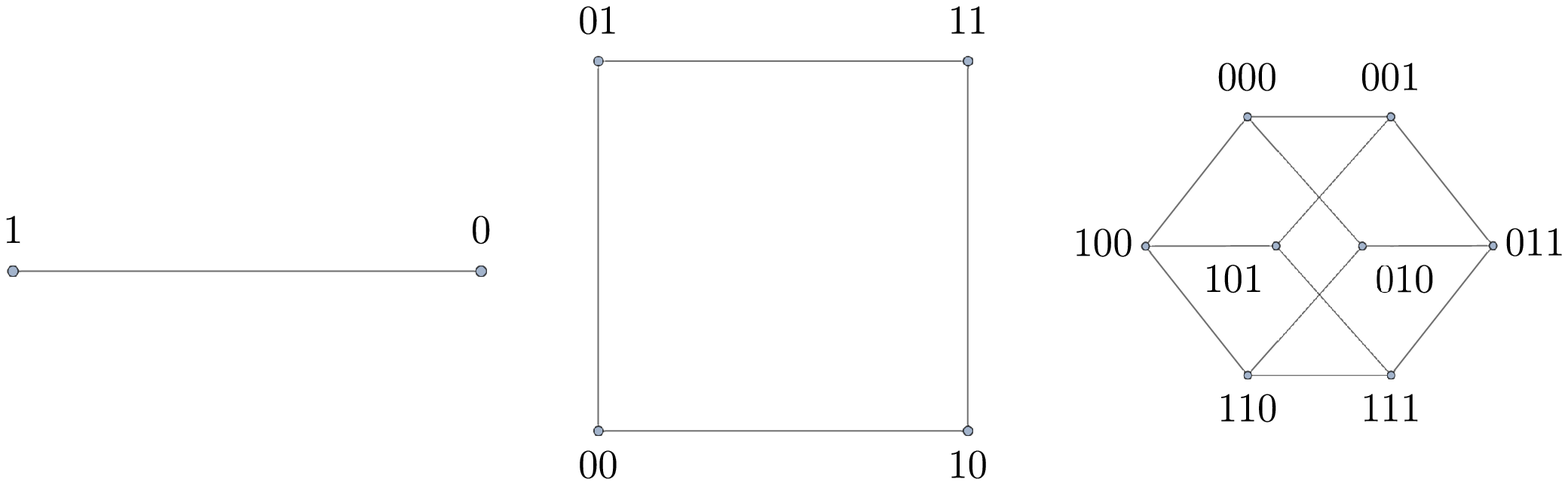}\\[5mm]
\includegraphics[width=0.5\textwidth]{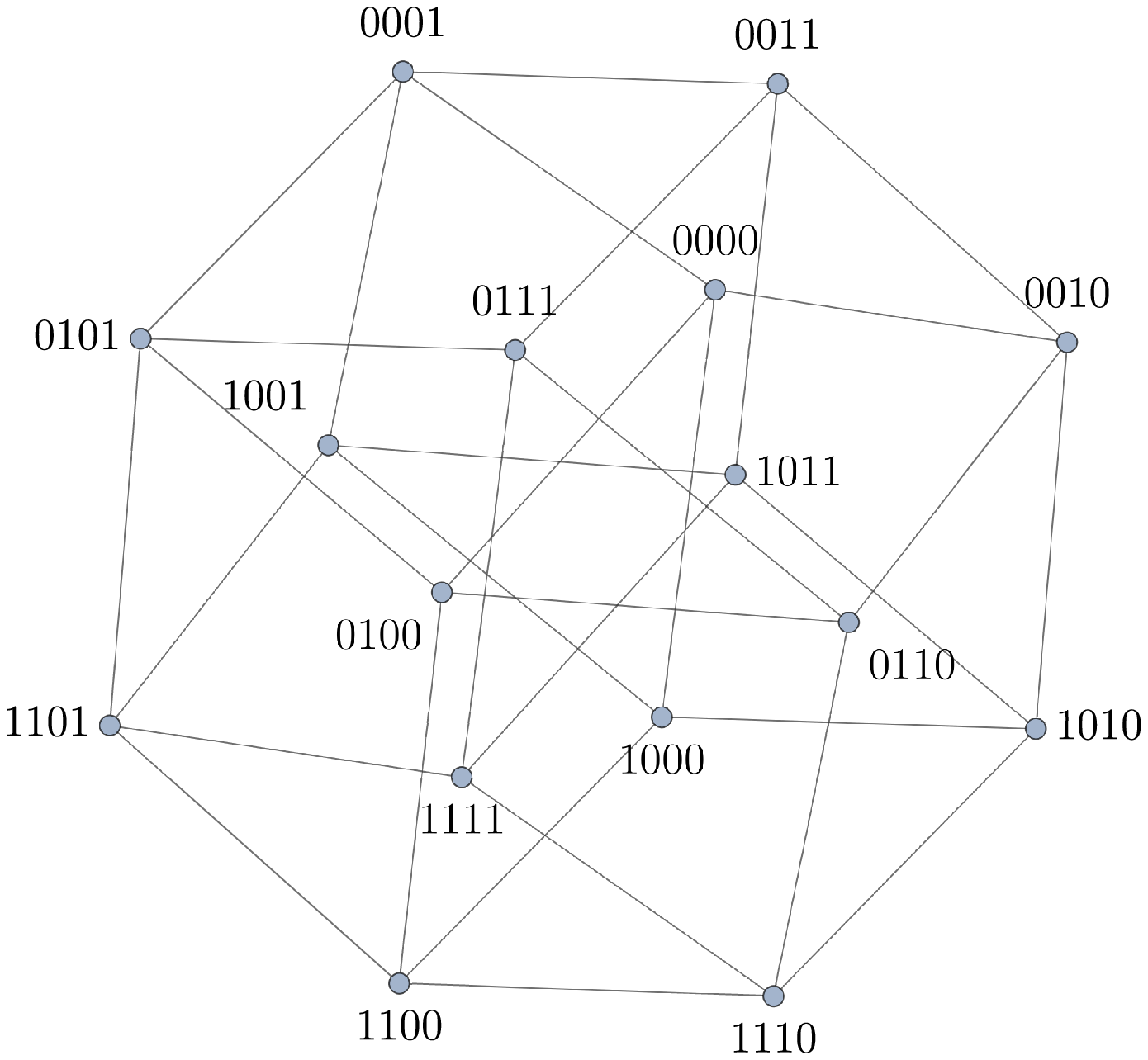}
\caption{The underlying geometry of the classical Eigen's model in
the dimensions $N=1,2,3$ and $4$. The vertices correspond to
different types of binary sequences, and the Hamming distance
between two sequences is given by the minimal number of edges
connecting them.}\label{fig:1}
\end{figure}

For matrix $\bs W$ it is possible to have different choices. One
of the most frequently used is the so-called \textit{single peaked
landscape} (SPL), which is defined as
$$
\bs W_{SPL}:=\diag(w+s,w,\ldots,w),\quad w\geq0,\,s>0.
$$
It turns out that it is impossible, however, to calculate $\overline{w}$ and $\bs p$ exactly in this case for finite values of $N$, and the first analysis of the quasispecies model with SPL relied heavily on numerical calculations (see \cite{swetina1982self} and Fig. \ref{fig:2}). Note that numerically it is not straightforward to solve the eigenvalue problem \eqref{eq0:1}, even for moderate values of $N$, because the dimension of the matrices is $2^N\times 2^N$. To overcome this difficulty, Swetina and Schuster \cite{swetina1982self} considered only the so-called \textit{permutation invariant fitness landscapes} whereas the fitness of a given sequence is determined by the Hamming distance from the master (zero) sequence. In this way one can track only the frequencies of class zero, which is the master sequence itself, of class one, which are all the sequences whose distance to the master sequence is one, etc, thus reducing the dimensionality of the problem to $(N+1)\times (N+1)$. SPL is an example of a permutation invariant fitness landscape.
\begin{figure}[!ht]
\centering
\includegraphics[width=0.48\textwidth]{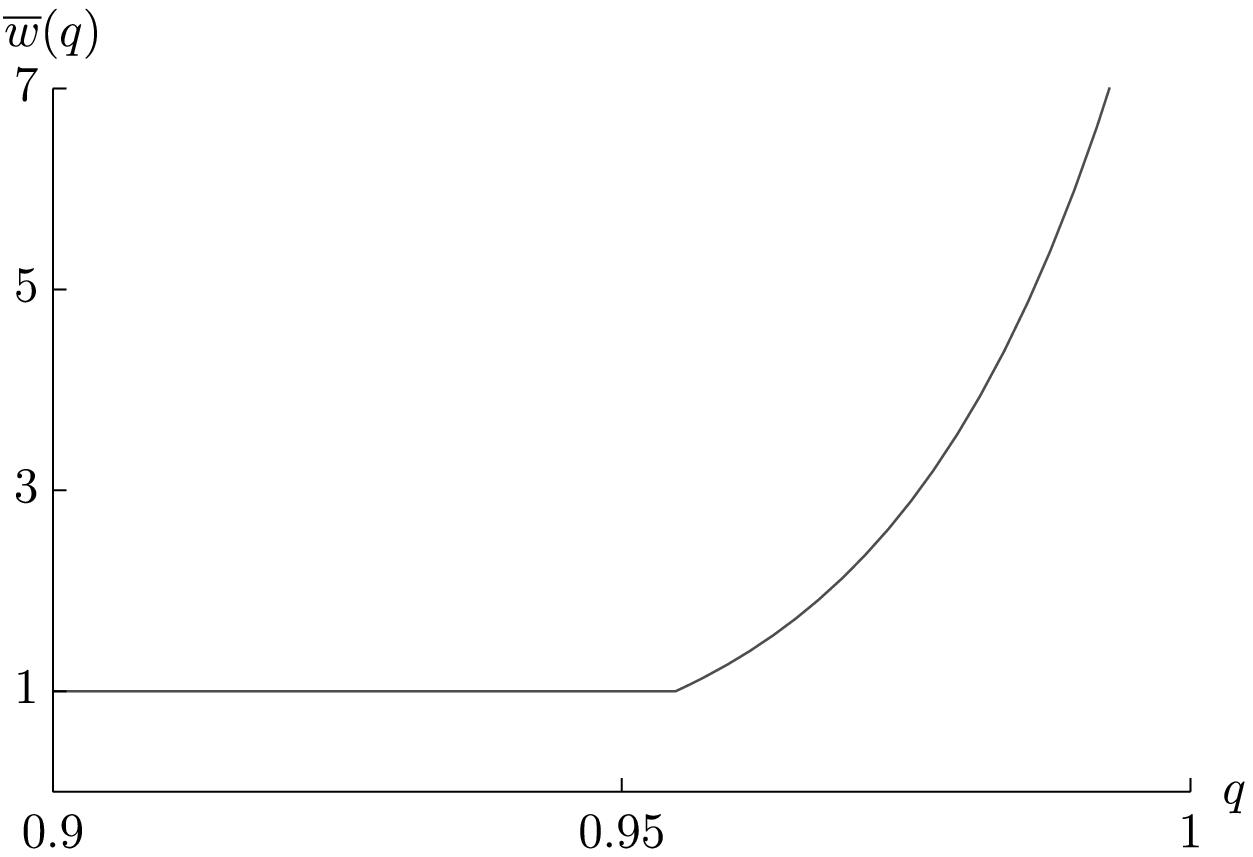}\hfill
\includegraphics[width=0.48\textwidth]{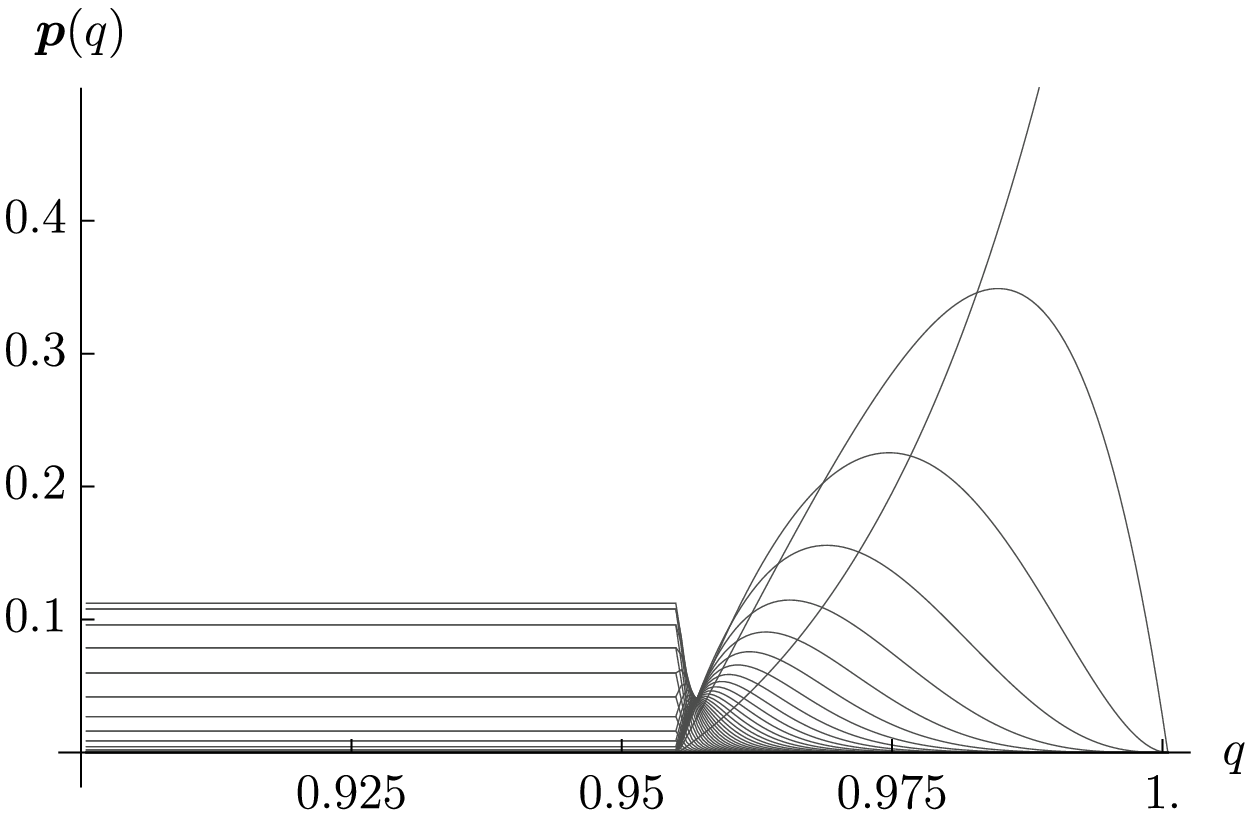}
\caption{Numerical solution of the quasispecies model with the single peaked landscape. $N=50$, $w=1$, $s=9$. On the left the leading eigenvalue is shown, on the right the $y$-axis gives the frequencies of the sequences with the same Hamming distance from the master sequence.}\label{fig:2}
\end{figure}

Fig. \ref{fig:2} shows the phenomenon of the notorious \textit{error threshold}: after some critical mutation rate the distribution of different types of sequences becomes uniform (and hence the distribution of classes shown in the right panel of Fig. \ref{fig:2} is binomial). We note that this phenomenon depends on the fitness landscape $\bs W$; for some $\bs W$ it does not manifest itself \cite{jainkrug2007,wiehe1997model,wilke2005quasispecies}.

It turns out that it is possible to exactly calculate $\overline{w}$ and $\bs p$ for Eigen's model in the case when the contributions to the overall fitness of different sites are independent \cite{dress1988evolution,rumschitzki1987spectral}, and the mathematical reason for this is the decomposition of the Eigen evolutionary matrix $\bs{QW}$ as
$$
\bs{QW}=\bs{Q}_0\bs{W}_1\otimes\bs{Q}_0\bs{W}_3\otimes \ldots\otimes \bs{Q}_0\bs{W}_N,
$$
where
$$
\bs Q_0=\begin{bmatrix}
          q & 1-q \\
          1-q & q \\
        \end{bmatrix},\quad \bs W_k=\begin{bmatrix}
                                      1 & 0 \\
                                      0 & s_k \\
                                    \end{bmatrix},\quad k=1,\ldots, N,
$$
$s_k$ is the contribution of the $k$-th site to the fitness, and
$\otimes$ is the Kronecker product. Biologically, this case
describes the absence of \textit{epistasis}. More generally, as it
was first noted in \cite{dress1988evolution}, the exact solution
is in principle can be given if the structure of matrix $\bs W$ is
related to the group of isometries of the binary hypercube
$X=\{0,1\}^N$ (see also below).

Around the same time (at the end of 1980s) another major
breakthrough about Eigen's model was achieved: It was shown that
the quasispecies model \eqref{eq0:1} is equivalent to the Ising
model of statistical physics
\cite{leuthausser1987statistical,leuthausser1986exact}, which
actually caused a stream of papers that used methods of
statistical physics to analyze \eqref{eq0:1} for various choices
of $\bs W$ (see \cite{baake2001mutation} and references therein).
Without going into the details (see, e.g.,
\cite{thompson2015mathematical} for an introduction to the Ising
model), we mention that the Ising model is formulated for a
given undirected graph, where the vertices can be in one of two
states, and the edges represent the interactions between the
vertices. In the classical two-dimensional Ising model that was
solved by Onsager in 1944 \cite{onsager1944crystal} the graph is
the lattice $\Z^2$. The solution is given in the limit when the
number of vertices approaches infinity, and originally was
obtained by analyzing the so-called \textit{transfer matrix},
which, as was shown in
\cite{leuthausser1987statistical,leuthausser1986exact} is exactly
equivalent to the evolutionary Eigen matrix $\bs{QW}$. Moreover,
the error threshold in Eigen's model is the \textit{phase
transition} in the Ising model.

Eventually the methods of statistical physics led to \textit{the maximum principle} for the quasispecies model \cite{Baake2007,Hermisson2002} (see also \cite{saakian2006ese}) that provides an efficient way of calculating the dominant eigenvalue $\overline{w}$ in the case of permutation invariant fitness landscapes and under some ``continuity'' condition on the limit of entries $\bs W$ when $N\to\infty$. Recently, the explicit expressions for the quasispecies distribution $\bs p$ for the permutation invariant fitness landscapes were obtained \cite{cerf2016quasispecies,cerf2016quasispeciesnew}.

Summarizing, we remark that, notwithstanding all the progress in the analysis of Eigen's model \eqref{eq0:1} outlined above, there are a great deal of open questions. In particular, we still lack analytical tools to tackle ``non-continuous'' fitness landscapes (but see \cite{semenov2016eigen}), most of the existing approaches work only with permutation invariant landscapes, and there exist no necessary and sufficient conditions for the existence of the error threshold, to mention just a few. Most importantly, from our point of view, the existing analysis of Eigen's model is almost exclusively concentrated on the case of the binary cube geometry (Fig. \ref{fig:1}), which is supported by the biological motivation for the model (because the RNA and DNA molecules are literally polynucleotide sequences). Mathematically, however, nothing precludes us from considering an abstract model on a finite metric space $X$ with some natural metric, thus changing the mutational landscape of  Eigen's model. We introduced such abstract model in \cite{semenov2016eigen} and the rest of the present paper is devoted to a detailed presentation of this model and its analysis.

\section{Generalized Eigen's model and algebraic background}\label{sec:3}
\subsection{Groups of isometries and a generalized algebraic Eigen's problem }\label{sec:3:1}
Let $(X,d)$ be a finite metric space. We will assume that the
metric $d\colon X\times X \longrightarrow \N_0$ is an {\it integer}-valued
function. Consider a group $\Gamma\leqslant{\rm Iso}(X)$ of
isometries of $X$ and suppose that $\Gamma$ acts {\it
transitively} on $X$, that is, $X$ is a single $\Gamma$-orbit (we
consider the left action).

Since $\Gamma$ acts transitively on $X$ we may fix an arbitrary
point $x_0\in X$ and consider the function $d_{x_0}\colon X\longrightarrow
\N_0$ such that $d_{x_0}(x)=d(x,x_0)$. By definition,
$$\diam X:=\max\{d_{x_0}(x)\mid x\in X\}$$
 is called the {\it diameter}
of $X$. The number $N=\diam X$ does not depend on the choice
of $x_0$.

Below we give two natural examples of such metric spaces. Arguably, the second example is mathematically more attractive, however, to keep a close connection to the classical Eigen's model discussed in Section \ref{sec:2}, the detailed calculations are presented for the metric spaces of more geometrically appealing Example \ref{ex1.1}.

\begin{example}[\textit{Regular polytopes}] \label{ex1.1} Let $X$ be the the set of
vertices of an $n$-dimensional regular polytope $P$ (see, e.g.,
\cite{coxeter1973regular}), all edges of which have an integer
length $e$. For example we can consider a regular $m$-gon ($m\ge
3$) on the plane, a tetrahedron, cube, octahedron, dodecahedron,
or icosahedron in the 3-dimensional space (see Fig. \ref{fig:3})
and so on, equipped with the ``edge'' metric: the distance between
$x$ and $y$ is the minimal number of edges of $P$ connecting $x$
and $y$ multiplied by $e$. For the $n$-dimensional unit cube the
edge metric is the same as the Hamming metric.

The full group of isometries $\Gamma={\rm Iso}(P)$  acts on $P$
and, consequently, on $X$. For instance, let $P$ be an icosahedron
or dodecahedron. Then $\Gamma\cong A_5$ where $A_5<S_5$ is the
alternating group of order 60.

\begin{figure}[!bh]
\centering
\includegraphics[width=0.19\textwidth]{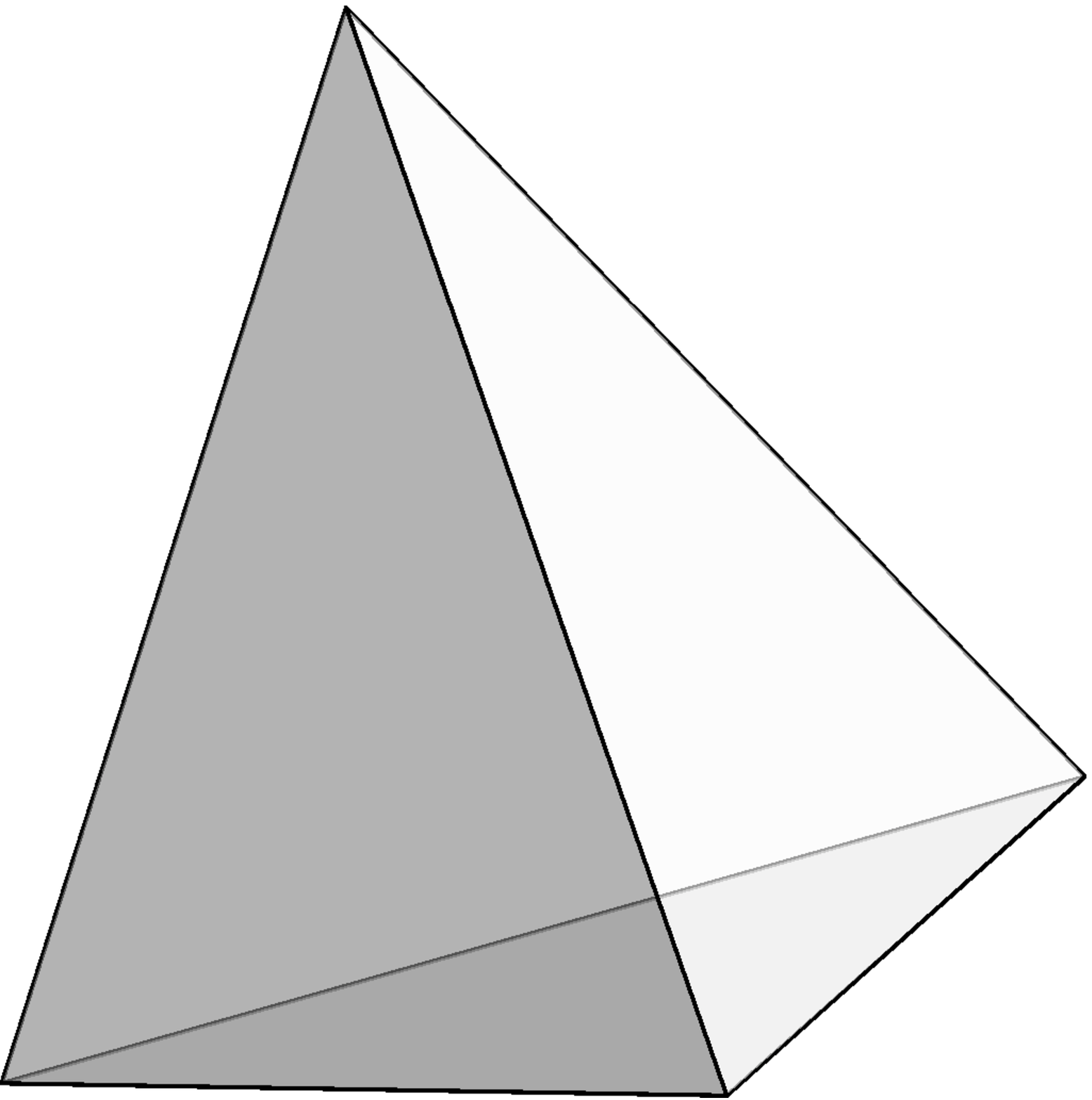}
\includegraphics[width=0.19\textwidth]{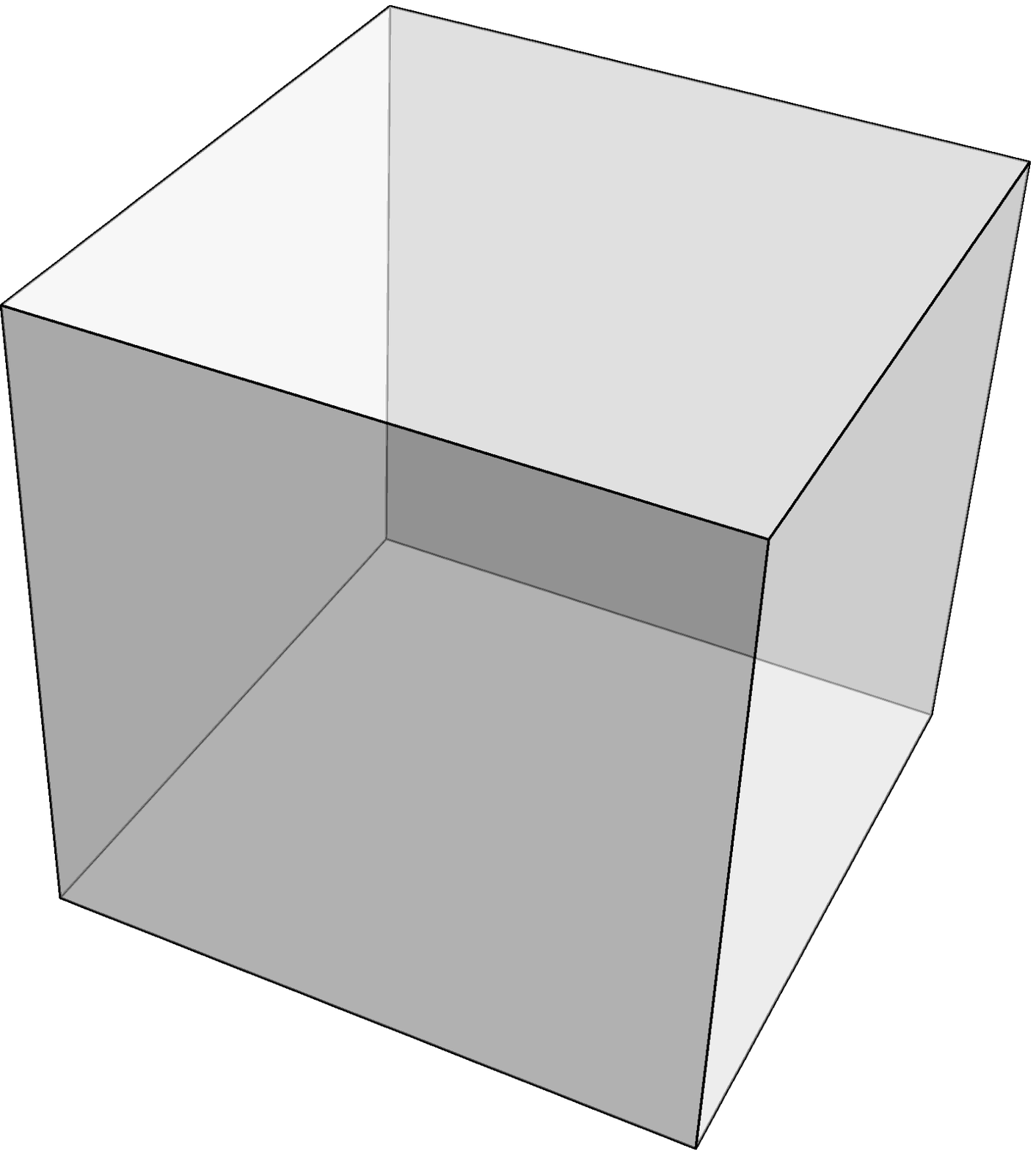}
\includegraphics[width=0.19\textwidth]{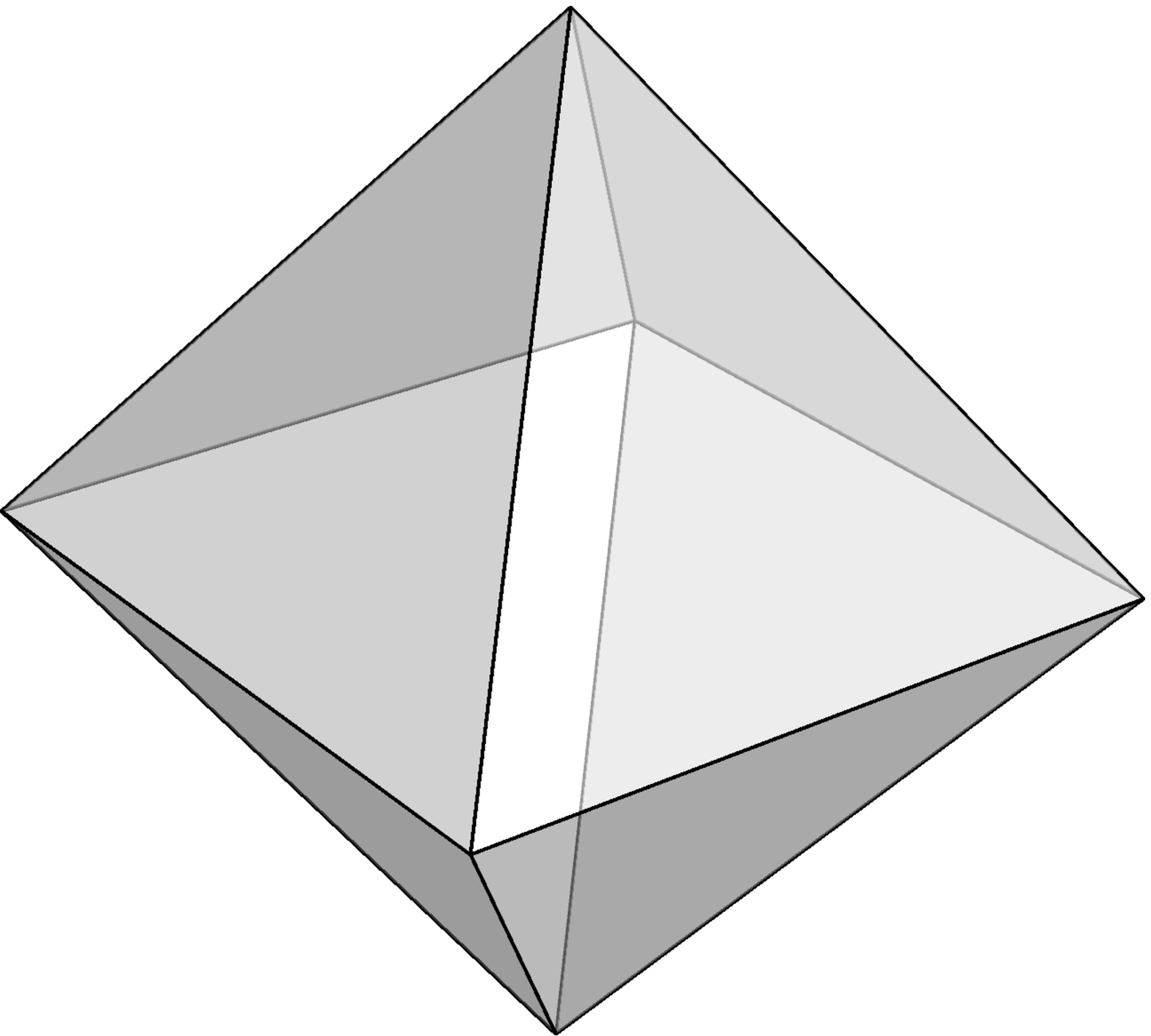}
\includegraphics[width=0.19\textwidth]{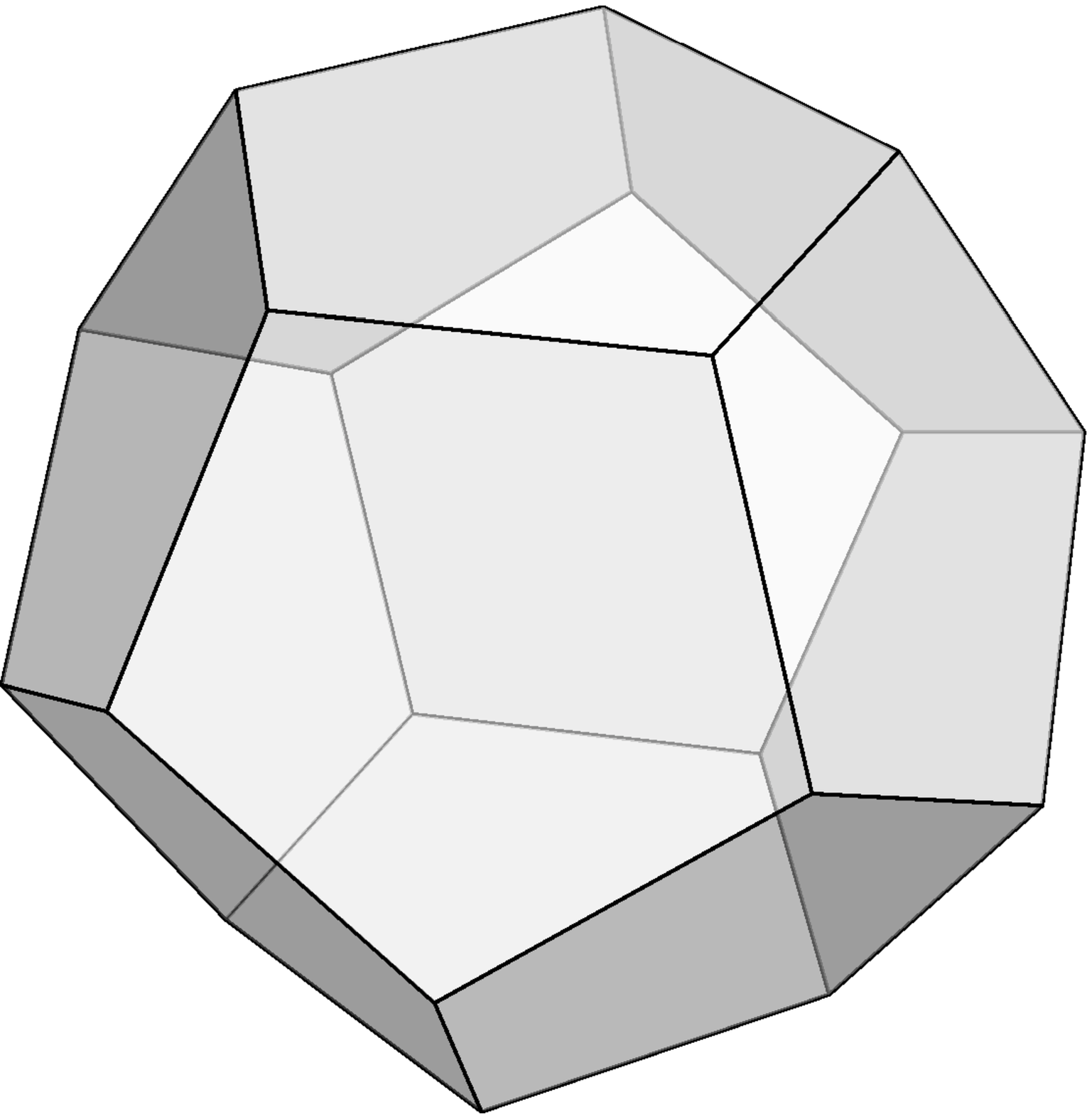}
\includegraphics[width=0.19\textwidth]{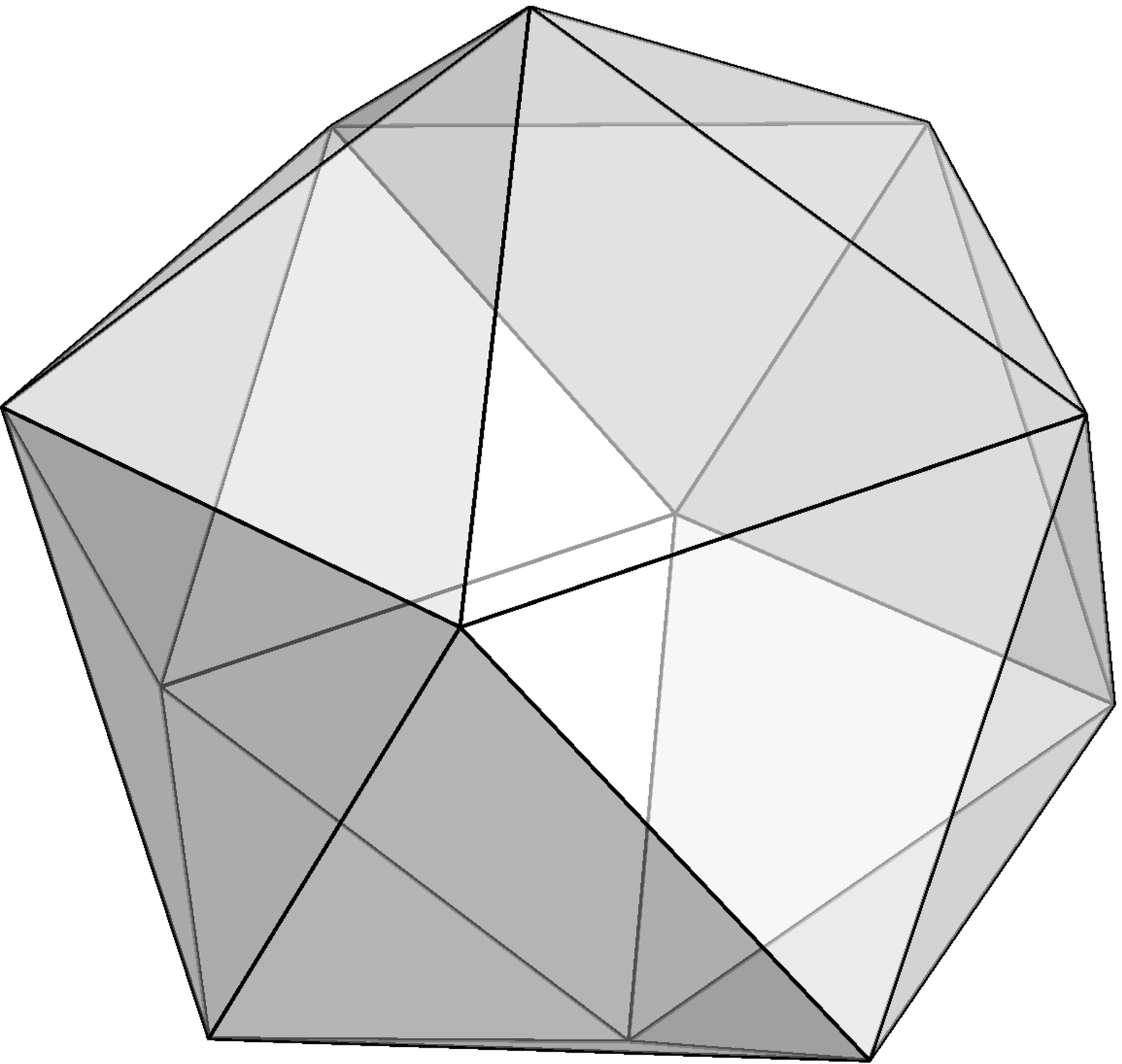}
\caption{Platonic solids (examples of regular polytopes) in dimension 3. From left to right: tetrahedron (regular simplex), cube, octahedron, dodecahedron, icosahedron.}\label{fig:3}
\end{figure}
\end{example}

\begin{example}[Groups as metric spaces]   Let $G$ be a finite group generated by a set
$S=S^{-1}$. The {\it word} metric $d=d_S$ on $G$ is defined as
follows (see \cite[chapter IV ]{de2000topics} for more details and
examples): $d(g,h)=l(g^{-1}h)$ where $l(g^{-1}h)=l$ is the minimal
number of generators $s\in S$ needed to represent $g^{-1}h$ as a
product $s_1\dots s_l$. The word metric is  invariant with respect
to the action of $G$ on itself by left shifts $h\mapsto gh$. Hence, we
have the metric space $X=G$ and the transitive action of
$\Gamma=G$ on $X$ by isometries.

More generally, for any subgroup $H<G$ we can define the metric
space $X_H=\{gH\,|\,g\in G\}$ of the left cosets of $G$ by $H$.
The group $G$ acts on $X_H$ by left shifts and
$$d(gH,aH)=\min\{d(x,y)\mid x\in gH,\, y\in aH\}\,.
$$
If $G$ acts transitively by isometries on a metric space $X$ then
as a $G$-set $X$ is isomorphic to the set of left cosets $G/{\rm
St}_{\Gamma}(x_0)$, $x_0\in X$, where ${\rm St}_{\Gamma}(x_0)$ is
the stabilizer (or the group of isotropy) of $x_0$ in $\Gamma$.
\end{example}

\medskip
Now  consider a quadruple $(X,d,\Gamma, {\bs w})$ where $(X,d)$ is
a finite metric space of
diameter $N$ with integer distances between points and cardinality $l=|X|$, a group $\Gamma\leqslant
{\rm Iso} (X)$ is a fixed group  and a {\it fitness function}
${\bs w}\colon X\longrightarrow \R_{\geq 0}$. The fitness function is often
represented by the vector-column ${\bs w}=(w_x)$ with non-negative
real entries called {\it fitnesses} which are indexed by $x\in X$
(for an appropriate ordering of $X$).

\begin{definition}\label{def:3:3} The quadruple $(X,d,\Gamma, {\bs w})$  is called {\it
homogeneous} $\Gamma$-landscape. It is called {\it symmetric} if for any points $y, z\in X$ there is an isometry
$\gamma=\gamma(y,z)\in \Gamma$ such that $\gamma y=z$ and $\gamma
z=y$.\end{definition}


%
%

Consider also the diagonal matrix ${\bs W}={\rm diag}(w_x)$ of
order $l$ called the {\it fitness matrix}, the symmetric distance
matrix ${\bs D}=\bigl(d(x,y)\bigr)$ of the same order with integer entries,
 and the symmetric matrix
${\bs Q}=\left((1-q)^{d(x,y)}q^{N-d(x,y)}\right)$ for $q\in
[0,1]$. Finally, we introduce the {\it distance polynomial}
\begin{equation}\label{1.1}
P_X(q)=\sum_{x\in X} (1-q)^{d(x,x_0)}q^{N-d(x,x_0)}\,,\quad
x_0\in X\,. \end{equation}
Since $\Gamma$ acts transitively on $X$ this polynomial is
independent of the choice of $x_0\in X$ and is the sum of entries
in each row (column)  of ${\bs Q}$.

\medskip The following key
definition generalizes the classical Eigen's problem.

\begin{definition} The  problem to find the leading eigenvalue
$\overline{w}=\overline{w}(q)$ of the matrix
$\frac{1}{P_X(q)}\,\bs{QW}$ and the eigenvector $\bs{\hat p}=\bs{\hat p}(q)$
satisfying
\begin{equation}\label{1.2}
{\bs Q}{\bs W}\bs{\hat{p}}=P_X(q)\overline{w}\, \bs{\hat p},\quad \hat p_x=\hat p_x(q)>
0,\quad\sum_{x\in X} \hat p_x(q)=1, \end{equation} will be called
{\it generalized algebraic quasispecies or Eigen's problem}.
\end{definition}

\medskip Note that in ({\ref{1.2}})
\begin{equation}\label{1.3} \overline{w}=\sum_{x\in X} w_x\hat p_x\,.
\end{equation} Due to the Perron--Frobenius theorem a solution of this
problem always exists. Also note that the uniform distribution
vector
\begin{equation}\label{1.4}
\bs{\hat p}=\frac{1}{|X|}(1,\dots,1)^\top=\frac{1}{l}(1,\dots,1)^\top
\end{equation} provides a solution of (\ref{1.2}) in the case
of constant fitnesses $w_x\equiv w>0$. By construction matrix $\frac{1}{P_X(q)}\,{\bs Q}$ is symmetric
and double stochastic. It will be called {\it generalized mutation
matrix}.

Problem (\ref{1.2}) turns into classical Eigen's quasispecies
problem if $X=\{0,1\}^N$ is the $N$-dimensional binary cube with
the Hamming metric $H$, and the group $\Gamma={\rm Iso} (X)$
named in 1930 by A. Young the {\it hyperoctahedral} group.
$\Gamma$ is isomorphic as an abstract group to the Weyl group of
the root system of type $B_N$ or $C_N$ and is acting on the binary
cube. In this case $P_X(q)\equiv 1$. The case when $X$ is
the set of vertices of an $N$-dimensional simplex with the
isometry group ${\rm Iso}(X)\cong S_{n+1}$ is treated in detail in
Section 6 of \cite{semenov2016eigen}. Here we continue with a
general analysis of the generalized quasispecies problem. The
first step is to study the properties of the distance polynomials.

%
%

\subsection{ Some general properties of the distance polynomial }
\label{subs1.2}

Using the notations of Section \ref{sec:3:1} we consider the
polynomial $P_X(q)=P_{X,d}(q)$. Polynomial $P_X(q)$ is strictly
positive on $[0,1]$ (provided the parameter $N$ is strictly equal
to $\diam X$)  and
 possesses the following properties, which are checked by direct calculations:

\begin{enumerate}
\item
\begin{equation}\label{1.7}
P_X(1)=1,\quad P_X\left(\frac{1}{2}\right)=
\frac{|X|}{2^N}=\frac{l}{2^N}\;. \end{equation}

\item
\begin{equation}\label{1.8}
P_X(q)=\sum_{k=0}^N f_k\,(1-q)^k q^{N-k}\in \Z[q]\,,
\end{equation} where the non-negative integers $f_k=f_k(X)=:\#\{x\in X\mid d(x,x_0)=k\}$ are the cardinalities of
$d$-spheres in $X$ with the center at the fixed point $x_0$ and of
radius $k$.
\end{enumerate}

\begin{remark}\label{rem:3:5} Polynomial $S_X(t)=\sum_{k=0}^N f_k
t^k$  is often called the {\it spherical growth function} of
$(X,d)$. See, for instance, \cite[chapter IV ]{de2000topics} for details
and examples.
\end{remark}

%

\subsection{An orbital ring associated with the triple $(X,d,\Gamma)$
}\label{subs1.3}

In this section, to study the spectral properties of the mutation matrix $\bs Q$, we introduce what we call \textit{an orbital ring}. For more algebraic details and construction of similar structures we
refer the reader to \cite{brown2012cohomology,feit1982representation,kirillov1976elements,semenov1994rings,serre2012linear}.

Specifically, let $(X,d,\Gamma, {\bs w})$
be a homogeneous symmetric (in the sense of Definition \ref{def:3:3})
$\Gamma$-landscape ($\Gamma\leqslant{\rm Iso}(X)$). We attach to the triple $(X,d,\Gamma)$ a commutative ring
$R=R(X,d,\Gamma)$ with unity, which we call the {\it orbital ring}. As an
abelian group $R$ is free and of rank ${\rm
rk}_{\Z}R=|\Gamma_0\backslash\Gamma/\Gamma_0|$, the number of
double $(\Gamma_0, \Gamma_0)$-cosets in $\Gamma$ where
$\Gamma_0={\rm St}_{\Gamma}(x_0)$.

Since $\Gamma$ acts by isometries on $X$ then the distance
function $d=d(x,y)$ is $\Gamma$-invariant with respect to the
diagonal action of $\Gamma$ on the cartesian square $X\times X$,
namely, $d(\gamma x,\gamma y)=d(x,y)$ for any $x,y \in X$ and
$\gamma\in\Gamma$.

Let $A$ be a $\Gamma$-orbit in $X\times X$ and let ${\bs M}_A$ be
the matrix with entries $({\bs M}_A)_{x,y}$  equal to 1 if
$(x,y)\in A$ and equal to 0 otherwise. It is worth mentioning that
${\bs M}_A$ can be identified with the matrix of
$\Gamma$-invariant $\Z$-linear endomorphism $f_A\in {\rm Hom}(\Z
X, \Z X)$ such that $f_A(y)=\sum_{(x,y)\in A} x$,  $\Z X$ being a
permutation $\Z\Gamma$-module (see, e.g., \cite{brown2012cohomology}).

It is well known that the set ${\rm Orb}$ of $\Gamma$-orbits $A$
in $X\times X$ is in 1-to-1 correspondence with the set
$\Gamma_0\backslash\Gamma/\Gamma_0$ of double $(\Gamma_0,
\Gamma_0)$-cosets in $\Gamma$ where $\Gamma_0={\rm
St}_{\Gamma}(x_0)$. We will say that $\Gamma$-orbit $A$ is of
degree $k$ ($\deg A =k$), $k=0,\dots, N$, if $d(x,y)=k$ for
some (and hence for any) $(x,y)\in A$. By definition, all $\Gamma$-orbits
$A$ of degree $k$ compose a subset ${\rm Orb}_k\subset {\rm Orb}$.

Note that the single $\Gamma$-orbit of degree 0 is the diagonal
$\Delta\subset X\times X$. The corresponding matrix ${\bs
M}_\Delta={\bs I}$, the identity matrix. Since different orbits
$A$ are disjoint the matrices ${\bs M}_A$ are independent over
$\Z$.

Therefore, we have the following expansion of the mutation matrix ${\bs Q}$:
\begin{equation}\label{1.9}
{\bs Q}= \sum_{k=0}^N (1-q)^k q^{N-k}\sum_{A\in {\rm Orb}_k} {\bs
M}_A\,,
\end{equation}
and the equality
\begin{equation}\label{1.10}
{\bs E}=\sum_{A\in {\rm Orb}} {\bs M}_A=\sum_{k=0}^N \sum_{A\in
{\rm Orb}_k} {\bs M}_A\,,
\end{equation}
where ${\bs E}$ is the matrix with all the entries equal to 1.

\begin{lemma} If the
triple $(X,d,\Gamma)$ is symmetric in the sense of Definition
\ref{def:3:3} then matrices ${\bs M}_A$ are symmetric and commute
pairwise. Moreover, there are integer non-negative structural
constants $\mu_{AB}^C=\mu_{BA}^C$, where $A, B, C \in {\rm Orb}$,
such that
\begin{equation}\label{1.11}
{\bs M}_A{\bs M}_B=\sum_{C\in {\rm Orb}} \mu_{AB}^C {\bs
M}_C\,.
\end{equation}
\end{lemma}
\begin{proof}
Let $(x,z)\in A$. In view of Definition \ref{def:3:3} there exists
an isometry $\gamma=\gamma(x,z)\in \Gamma$ such that $\gamma x=z$
and $\gamma z=x$. Thus, $(z,x)\in A$ and ${\bs M}_A$ is symmetric.

Moreover, it follows from the definition that for any $x, z \in X$
the corresponding matrix entry
\begin{equation}\label{1.12}
({\bs M}_A{\bs M}_B)_{x,z}=\#\{y\in X\mid(x,y)\in
A,\;(y,z)\in B\}\,.
\end{equation}

On the other hand, for the same transposing isometry
$\gamma=\gamma(x,z)\in \Gamma$
$$
\#\{y\in X\mid(x,y)\in A,\;(y,z)\in B\}= \#\{y\in X\mid(z,y)\in
B,\;(y,x)\in A\}= $$
$$=\#\{\gamma y\in X\mid(\gamma z, \gamma y)\in
B,\;(\gamma y, \gamma x)\in A\}=\#\{\gamma y\in X\mid(x, \gamma
y)\in B,\;(\gamma y, z)\in A\}\,.
$$
It follows that $({\bs M}_A{\bs M}_B)_{x,z}=({\bs
M}_B{\bs M}_A)_{x,z}$ and ${\bs M}_A{\bs M}_B={\bs
M}_B{\bs M}_A$.

The pair $(x,z)$ defines a $\Gamma$-orbit $C$. For $g \in\Gamma$
we have the same as in (\ref{1.12}) non-negative number
$$
({\bs M}_A{\bs M}_B)_{g x, g z}=\#\{gy\in X\mid(g x, gy)\in
A,\;(g y, g z)\in B\}\,.
$$
Hence, (\ref{1.11}) holds for some non-negative integer constants
$\mu_{AB}^C$.
 The lemma is proved.
\end{proof}

Consequently, we have proved

\begin{theorem}
 All  $\Z$-linear combinations of $ {\bs M}_A$, $A\in {\rm Orb}$,
compose a commutative unital ring
$R=R(X,d,\Gamma)$ with unity ${\bs M}_\Delta={\bs I}$ called the
{\it orbital ring} associated with the symmetric triple
$(X,d,\Gamma)$. As $\Z$-module $R$ is free of rank ${\rm
rk}_{\Z}R=|{\rm Orb}|=|\Gamma_0\backslash\Gamma/\Gamma_0|$.
\end{theorem}

\begin{example}\label{ex1.8} Let $X=\{0,1\}^2$ be the binary square with points
$x_0=[0,0]$, $x_1=[0,1]$, $x_2=[1,0]$, $x_3=[3,1]$ (binary
representation of indices) with the Hamming metric $d$. Let
$\Gamma\cong D_4$ (the dihedral group of order 8) be the group of
all isometries of $X$. Then the triple $(X,d,\Gamma)$ is
symmetric.

The set ${\rm Orb}$ consists of three orbits (corresponding to the
three orbits, namely, $d$-spheres $\{x_0\}=S_0(x_0)$, $\{x_1,
x_2\}=S_1(x_0)$, $\{x_3\}=S_2(x_0)$, of the stabilizer
$\Gamma_0={\rm St}_{\Gamma}(x_0)\cong \Z/2\Z$ acting on $X$)
represented by matrices
$$
{\bs M_0}={\bs I}=\left[
\begin{array}{cccc}
1&0&0&0\\
0&1&0&0\\
0&0&1&0\\
0&0&0&1\\
\end{array}\right]\,,\quad {\bs M_1}=\left[
\begin{array}{cccc}
0&1&1&0\\
1&0&0&1\\
1&0&0&1\\
0&1&1&0\\
\end{array}\right]\,,\quad
{\bs M_2}=\left[
\begin{array}{cccc}
0&0&0&1\\
0&0&1&0\\
0&1&0&0\\
1&0&0&0\\
\end{array}\right]\,.
$$
\end{example}
The multiplication table of these matrices in $R=R(X,d,\Gamma)$ is
as follows:
$$
\begin{array}{|c||c|c|c|}\hline
\times&{\bs I}&{\bs M_1}&{\bs M_2}\\ \hline {\bs I}&{\bs I}&{\bs
M_1}&{\bs M_2}\\ \hline {\bs M_1}&{\bs M_1}&2{\bs I}+2{\bs
M_2}&{\bs M_1}\\\hline {\bs M_2}&{\bs M_2}&{\bs M_1}&{\bs
I}\\\hline
\end{array}
$$

\begin{remark} We have not yet applied the triangle
inequality $d(x,z)\leq d(x,y)+d(y,z)$. It can be used for the
construction of a graded ring ${\rm gr}\, R={\rm gr}\, R(X,d,\Gamma)$.
Consider the following increasing filtration on $R$:
$$
R_{-1}=0<\Z=R_0<R_1<\dots<R_N=R\,\quad \mbox{for
$\Z$-modules}\;\;R_k=\bigoplus_{{\rm deg} A\leq k}\Z{\bs M}_A\,.
$$
It follows from the definition and the triangle inequality that
$R_i\cdot R_j\subseteq R_{i+j}$. Hence, we can attach to the
triple $(X,d,\Gamma)$ the  graded ring
$${\rm gr}\, R=\bigoplus_{ k=0}^{N} R_k/R_{k-1}\,. $$

For instance, in the above Example \ref{ex1.8} (here ${\bs M}_k\in
{\rm gr}_k R=R_k/R_{k-1}$ is viewed as the corresponding element
of $R_k$ modulo $R_{k-1}$):

$$
\begin{array}{|c||c|c|c|}\hline
\times&{\bs I}&{\bs M_1}&{\bs M_2}\\ \hline {\bs I}&{\bs I}&{\bs
M_1}&{\bs M_2}\\ \hline {\bs M_1}&{\bs M_1}&2{\bs M_2}&0\\\hline
{\bs M_2}&{\bs M_2}&0&0\\\hline
\end{array}
$$
\end{remark}

\subsection{Spectral properties of the mutation matrix $\bs Q$}

Consider now the space $V={\rm Hom}_{\R}(X,\R)$ of all linear
functions $f\colon X\longrightarrow \bf R$. Each function of $V$ can be represented
as a vector-column ${\bs v}=(f(x))$ (in fact, a {\it covector}). The
matrix ${\bs M}_A$ (see the previous section for the definition)  defines a linear endomorphism ${\bs M}_A: V\longrightarrow
V$ such that
$${\bs M}_A f(x)=\sum_{y\colon (x,y)\in A} f(y)\,.$$
If ${\bs v}=(f(x))$ then this endomorphism is just the
multiplication ${\bs v}\mapsto {\bs M}_A{\bs v}$.

Let us show that  each endomorphism ${\bs M}_A$ commutes with
$\Gamma$-action on $V={\rm Hom}_{\R}(X,\R)$ given by the rule
$\gamma f(x)=f(\gamma^{-1}x)$, $\gamma\in \Gamma$. In fact,
\begin{equation}\label{1.13}
\gamma{\bs M}_A f(x)= {\bs M}_A f(\gamma^{-1}x)=
\sum_{y:\,(\gamma^{-1}x,y)\in A}\!\!\! f(y)=\sum_{y:\,(x,\gamma
y)\in A}\!\!\! f(y)=\sum_{z:\,(x,z)\in A}\!\!\!
f(\gamma^{-1}z)={\bs M}_A \gamma f(x)\,.
\end{equation}

\begin{theorem} \label{thm1.9} Let the
triple   $(X,d,\Gamma)$ be symmetric. Then there exists a
non-degenerate real constant transition
matrix ${\bs T}=(t_{x,y})$ of order $l=|X|$ such that
\begin{enumerate}
\item All matrix entries of ${\bs T}$ are integer algebraic
\emph{(}over the field $\Q$\emph{)}
numbers.\\
\item The column of ${\bs T}$ indexed by a fixed  $x_0\in X$ is
equal to
$\bs 1=(1,\dots,1)^\top$. If $y\ne x_0$ then $\sum_{x\in X}t_{x,y} =0$.\\
\item
\begin{equation}\label{1.14}
{\bs T}^{-1}{\bs Q}{\bs T}=\diag(P_x(q))\,,
\end{equation}
where $P_{x_0}(q)$ is the distance polynomial $P_X(q)\in \Z[q]$,
other eigenpolynomials $P_x(q)\in \R[q]$ have integer algebraic
coefficients and $P_x(1)=1$ for any $x\in X$.\\
\item There are at most $r={\rm rk}_{\Z}R=|{\rm Orb}|$
different eigenpolynomials $P_x(q)$ in (\ref{1.14}).\\
\item For the distance matrix ${\bs D}$ and $N=\diam X$
$$
{\bs T}^{-1}{\bs D}{\bs T}=\diag(\lambda_x)\,,\qquad
\lambda_x=N
2^{N-1}P_x\left(1/2\right)-2^{N-2}P'_x\left(1/2\right)\,.
$$
\end{enumerate}
\end{theorem}

\begin{proof}Consider the space $V={\rm Hom}_{\R}(X,\R) $ of
linear functions $f\colon X\longrightarrow \bf R$. It is well known that each
symmetric matrix over $\R$ is diagonalizable, has real eigenvalues
and the family of commuting symmetric matrices ${\bs M}_A$, $A\in
{\rm Orb} $, has a common eigenbasis. In fact, if
$V_{\lambda}={\rm Ker}({\bs M}_A-\lambda {\bs I})$ is the eigenspace
corresponding to eigenvalue $\lambda$ of ${\bs M}_A$ then
$V_{\lambda}$ is ${\bs M}_B$-invariant:
$$
{\bs M}_A{\bs M}_B {\bs v}={\bs M}_B{\bs M}_A {\bs v}=\lambda{\bs
M}_A{\bs v}\,.
$$
Further, by induction on $r=|{\rm Orb}|$ we can conclude that
there is a common eigenbasis for all matrices~${\bs M}_A$.

1. First, since all the entries of ${\bs M}_A$ are
zeroes and ones, all the real eigenvalues are integer algebraic
numbers. Hence we can choose eigenvectors (vector-columns) ${\bs
t}_y$, $y\in X$, in a common eigenbasis with integer algebraic
entries (scaling the eigenvectors by appropriate integer factors
if necessary) and compose a transition matrix ${\bs T}$. Thus, the
first assertion is proved.

2. Moreover, the vector $\bs 1=(1,\dots,1)^\top$ (the
constant function $f(x)\equiv 1$ of $V$) is an eigenvector for
each ${\bs M}_A$ since
\begin{equation}\label{1.15}
{\bs M}_A \bs 1=s_A \bs 1,\quad s_A=
\#\{a\in X\mid (x_0,a)\in A\}.
\end{equation}
 Note that the leading eigenvalue $s_A$ of ${\bs M}_A$ is the cardinality of the $\Gamma_0$-orbit in $X$,
 $\Gamma_0={\rm St}_\Gamma(x_0)$, corresponding to $\Gamma$-orbit $A$ in $X\times X$.
 Each row (column) of ${\bs M}_A$ contains exactly $s_A$ ones.
  We may index the eigenvector $\bs 1$
by $x_0$.

The subspace $V_0=\{f\in V={\rm Hom}_{\R}(X,\R)\mid\sum_{x\in X}
f(x)=0\}$ is invariant for each ${\bs M}_A$. Hence,  we can choose the other
eigenvectors from this subspace.

3. Applying the conjugation by ${\bs T}$ to
(\ref{1.9}) we get
\begin{equation}\label{1.16}
{\bs T}^{-1}{\bs Q}{\bs T}=  \sum_{k=0}^N (1-q)^k
q^{N-k}\sum_{A\in {\rm Orb}_k} {\bs T}^{-1}{\bs M}_A{\bs
T}=\diag(P_x(q))\,
\end{equation}
for appropriate polynomials $P_x(q)$. Each $P_x(q)$ is a linear
combination:
$$
P_x(q)=\sum_{k=0}^N (1-q)^k q^{N-k}\sum_{A\in {\rm
Orb}_k}\lambda_{x,A} , \quad \diag(\lambda_{x,A})={\bs
T}^{-1}{\bs M}_A{\bs T}\,.
$$
Since $\lambda_{x_0,A}=s_A$ (see \ref{1.15}) then
$$P_{x_0}(q)=\sum_{k=0}^N (1-q)^k q^{N-k} \sum_{A\in {\rm
Orb}_k}s_A= \sum_{k=0}^N (1-q)^k q^{N-k}|S_{k}(x_0)|
$$ is the distance polynomial. For the definition of $S_k(x_0)$ see Remark \ref{rem:3:5}. For $q=1$ we have ${\bs Q}(1)={\bs
I}$ and the third assertion is proved.

4. Consider the subspace $V_{\Gamma_0}\subset
V={\rm Hom}_{\R}(X,\R)$ of all functions ${\bs v}=(f(x))$ which
are constant on the $\Gamma_0$-orbits in $X$, that is,
$\Gamma_0$-invariant functions in $V={\rm Hom}_{\R}(X,\R)$. In
view of (\ref{1.13}) the subspace $V_{\Gamma_0}={\rm
Hom}_{\Gamma_0}(X,\R)$  is ${\bs M}_A$-invariant for each $A\in
{\rm Orb}$.

It follows that we have $r={\rm rk}_{\Z}R$-dimensional
representation $\pi\colon R\longrightarrow {\rm End}_{\R}(V_{\Gamma_0})$ such that
$\pi\colon {\bs M}_A\longrightarrow {\bf M}_A$, since by construction the ring $R$ is the $\Z$-linear
span of ${\bs M}_A$. It is not hard to see that the representation
is exact (consider function $f$ such that $f(x_0)=1$ and
$f(x)=0$ if $x\ne x_0$). In what follows we identify the elements
of $\pi(R)$ with the corresponding matrices and
$\Gamma_0$-invariant functions $f(x)$ with the corresponding vectors
${\bs v}=(f(x))$.

Since the ring $R$ is commutative we can find, similar to the proofs of 1 and 3,
$\Gamma_0$-invariant eigenfunctions ${\bs v}_0=(P_0(x))$, \dots,
${\bs v}_{r-1}=(P_{r-1}(x))$ such that $P_0(x)\equiv 1$ and
$$
{\bs Q} {\bs v}_j=P_j(q){\bs v}_j\,,\quad j=0,\dots, r-1\,,
$$
where each $P_j(q)$ is an eigenpolynomial of ${\bs Q}$ coinciding
with one of $P_x(q)$ in (\ref {1.14}). In fact, let ${\bs t}_x$ be
the $x$-column of the transition matrix ${\bs T}$. Then ${\bs
Q}\,{\bs t}_x=P_x(q){\bs t}_x$. If ${\bs t}_x$ corresponds to a
$\Gamma_0$-invariant function in $V_{\Gamma_0}$ the proof is
complete. Otherwise we can suppose that the $y$-component
$t_{y,x}$ of ${\bs t}_x$ is nontrivial and apply the operation of
averaging
$$
{\bs t}_x\mapsto\sum_{\gamma\in \Gamma_y}\gamma{\bs t}_x ={\bs
t}\,,\qquad \Gamma_y={\rm
St}_{\Gamma}(y)=g\Gamma_0g^{-1},\;y=gx_0.
$$
Note that ${\bs t}$ corresponds to a $\Gamma_y$-invariant function
and $y$-component ${\bs t}$ is  equal to $|\Gamma_y|\,t_{y,x}$,
i.e., non-trivial. The representation $\pi_y\colon R\longrightarrow {\rm
End}_{\R}(V_{\Gamma_y})$ is equivalent to  $\pi\colon R\longrightarrow {\rm
End}_{\R}(V_{\Gamma_0})$ since $\Gamma$ acts transitively on $X$.
In view of (\ref{1.13}) each ${\gamma{\bs t}_x}$ is a common
eigenvector of all ${\bs M}_A$ and, consequently, of ${\bs Q}$,
corresponding to the eigenvalue $P_x(q)$, so is ${\bs t}$. Since
${\bs t}\ne {\bs 0}$ the proof is complete.

5. Finally, direct calculation yields the equality
$$
{\bs Q}'\left(1/2\right)=2N{\bs
Q}\left(1/2\right)-\frac{1}{2^{N-2}}{\bs D}\,,
$$
whence
$${\bs D}=N 2^{N-1}{\bs Q}\left(1/2\right)-
2^{N-2}{\bs Q}'\left(1/2\right)\,$$ and $$ {\bs T}^{-1}{\bs D}{\bs
T}=N 2^{N-1}{\bs T}^{-1}{\bs Q}\left(1/2\right){\bs T}-
2^{N-2}{\bs T}^{-1}{\bs Q}'\left(1/2\right){\bs
T}\,=\diag(\lambda_x)\,,$$ $\lambda_x=N
2^{N-1}P_x\left(1/2\right)-2^{N-2}P'_x\left(1/2\right)$.

By virtue of Assertion 4 there are at most $r={\rm
rk}_{\Z}R$ different eigenvalues of ${\bs D}$.

The theorem is
proved.
\end{proof}

\begin{example}\label{ex:3:11} In \cite{semenov2016eigen} we considered the quasispecies symmetric triple
$(X,d,\Gamma)$ where $X=\{0,1\}^n$ is the binary hypercube with
the Hamming metric $d$ of dimension $n=N= \diam X$   and $\Gamma$
is the hyperoctahedral group of order $2^n\cdot n!$.

Let $x_0=0=[0,\dots,0]\in X$ (the binary representation). Then
$\Gamma_0={\rm St}_{\Gamma}(x_0)\cong S_N$ and there are exactly
$r=N+1$ $\Gamma_0$-orbits $A_0$, \dots, $A_N$ in $X$, namely the
spheres $A_k=S_k(x_0)$ of cardinalities ${\binom{N}{k}}$. For
${\bs Q}$ there are exactly $r=N+1$ different eigenpolynomials
$P_k(q)=(2q-1)^k$, $k=0,\dots, N$, of multiplicities
$|A_k|={\binom{N}{k}}$.

We also considered the simplicial symmetric triple $(X,d,\Gamma)$
where $X$ is the 0-skeleton of the regular simplex such that
$|X|=n+1$ with unit distances between different vertices. The
group $\Gamma\cong S_{n+1}$ and $\Gamma_0\cong S_{n}$.

There are exactly $r=2$ $\Gamma_0$-orbits $A_0=\{x_0\}$,
$A_1=X\setminus A_0$ in $X$, namely, the spheres $A_0=S_0(x_0)$,
$A_1=S_1(x_0)$ of cardinalities 1, $n$. For ${\bs Q}$ there are
$r=2$ different eigenpolynomials, namely, $P_0(q)=q+n(1-q)$ of
multiplicity 1, $P_1(q)=2q-1$  of multiplicity $n$.
\end{example}

Together with the calculations we present below and summarized in a table form in Appendix \ref{ap:1} Example \ref{ex:3:11} prompts us to formulate the following conjecture.

\begin{conjecture}\label{conj1.11}
For an arbitrary orbital ring $R=(X,d,\Gamma)$ the eigenpolynomials of the corresponding mutation matrix ${\bs Q}$  can be
enumerated by $A\in {\rm Orb}$ and there are  exactly $r={\rm
rk}_{\Z}R$ different eigenpolynomials $P_A(q)$ of ${\bs Q}$ of
multiplicities $m_A$. It follows that
$$
\sum_{A\in {\rm Orb}} m_A =l=|X|\,.
$$
In addition, matrix ${\bs T}$ in Theorem \ref{thm1.9} can be
chosen to be symmetric.
\end{conjecture}

\section {$G$-invariant homogeneous symmetric $\Gamma$-landscapes, $G\leqslant \Gamma$}\label{sec:4}
Having at our disposal the orbital ring associated with the triple $(X,d,\Gamma)$ and, correspondingly, the spectral properties of $\bs Q$, we are in position to consider the eigenvalue problem \eqref{1.2}. To make progress we restrict ourselves to some special fitness landscapes, which are constant along $G$-orbits, where $G\leqslant \Gamma$.
\subsection{Reduced problem}\label{sec:4:1} Let $(X,d,\Gamma, {\bs w})$ be a homogeneous
$\Gamma$-landscape ($\Gamma\leqslant{\rm Iso}(X)$) and let
$G\leqslant \Gamma$ be a fixed subgroup. If $A$ is a $G$-orbit then $(A,d)$ is a metric subspace
of $(X,d)$ on which $G$ acts transitively by isometries. Consider
the restriction ${\bs w}|_A$. Thus, the quadruple $(A,d,G, {\bs
w}|_A)$ can be viewed as a homogeneous $G$-sublandscape of
$(X,d,\Gamma, {\bs w})$.

\begin{definition}We call $\Gamma$-landscape $(X,d,\Gamma, {\bs w})$
$G$-invariant if the fitness function ${\bs w}$ is constant on
each $G$-orbit $A$ of $G$-action on $X$, that is, ${\bs
w}(A)\equiv w_A\geq 0$.
\end{definition}
For instance, for the trivial subgroup $G=\{1\}$ each
homogeneous $\Gamma$-landscape is $G$-invariant.


Let a $\Gamma$-landscape $(X,d,\Gamma, {\bs w})$ be symmetric
and $G$-invariant. We suppose that $G$-invariant fitness
function ${\bs w}$ has at least two values. We will also assume
that there is a decomposition
\begin{equation}{\label{2.2}} X=A_0\sqcup\bigsqcup_{i=1}^t A_i\,
\end{equation}
such that $A_0$ is a union of $G$-orbits on which ${\bs
w}(A_0)\equiv w\geq 0$, and each $A_i$, $i=1,\dots, t$, is just a
single $G$-orbit  on which ${\bs w}(A_i)\equiv w+s_i$, where
$s_i>0$ ($s_i$ are not necessarily different). Then fitness matrix ${\bs W}$ can be
represented as follows
\begin{equation}{\label{2.3}}
{\bs W}= w{\bs I}+\sum_{i=1}^t s_i{\bs E}_{A_i}\;,
\end{equation}
$\bs{ I}$ being the identity matrix and ${\bs E}_{A_i}$ being the
projection matrix   with the  only nontrivial entries $e_{aa}=1$,
$a\in A_i$, on the main diagonal.

We want to solve problem (\ref{1.2}). In view of (\ref{2.3})
equation (\ref{1.2}) reads
$$
w\bs{ Q \hat p}+{\bs Q}\sum_{i=1}^t s_i{\bs
E}_{A_i}\bs{\hat p}=P_X(q)\overline{w}\bs{\hat p},
$$
whence
\begin{equation}{\label{2.4}}
(\overline{w}P_X(q){\bs I}-w{\bs Q})\bs{\hat p}={\bs Q}\sum_{i=1}^t
s_i{\bs E}_{A_i}\bs{\hat p}.
\end{equation}

For the matrix 1-norm we have $\| w{\bs
Q}\|_1=P_X(q)w<P_X(q)\overline{w}=\|
P_X(q)\overline{w}{\bs I}\|_1$. Consequently, the matrix
$P_X(q)\overline{w}{\bs I}-w{\bs Q}$ is non-singular and we
obtain the equality
$$
\bs{\hat p}=(P_X(q)\overline{w}{\bs I}-w{\bs Q})^{-1}\,{\bs Q}\bs v,\quad
\bs v=\sum_{i=1}^t s_i{\bs E}_{A_i}\bs{\hat p}.
$$
Multiplying the last equality by $\sum_{i=1}^t s_i{\bs E}_{A_i}$
yields
\begin{equation}{\label{2.5}}
\bs v=\sum_{i=1}^t s_i{\bs E}_{A_i}\bs{\hat p}=\sum_{i=1}^t s_i{\bs
E}_{A_i}(\overline{w}P_X(q){\bs I}-w{\bs Q})^{-1}\,{\bs Q}\bs v.
\end{equation}

Denoting \begin{equation}{\label{2.6}} {\bs M}=\sum_{i=1}^t s_i{\bs
E}_{A_i}(\overline{w}P_X(q){\bs I}-w{\bs Q})^{-1}\,{\bs Q}\,
\end{equation}
we can rewrite (\ref{2.5}) as
\begin{equation}{\label{2.7}}
\bs v={\bs M}\bs v\,,\quad \bs v=\sum_{i=1}^t s_i{\bs E}_{A_i}\bs{\hat p},
\end{equation}
 and hence vector $\bs v$ is an eigenvector of ${\bs M}$
corresponding to the eigenvalue $\lambda=1$.

Considering $\overline{w}$ in (\ref{2.6}), (\ref{2.7}) as a parameter we now concentrate on the following \textit{reduced problem}: To find the eigenvector
$\bs v$ satisfying (\ref{2.7}) and corresponding to the eigenvalue
$\lambda=1$ of matrix ${\bs M}$ defined in (\ref{2.6}).

\begin{remark} Expanding the  right-hand side of (\ref{2.6}) we get
\begin{equation}{\label{2.8}}{\bs
M}=\frac{1}{\overline{w}P_X(q)}\sum_{i=1}^t s_i
\sum_{m=0}^\infty\left(\frac{w}{\overline{w}P_X(q)}\right)^m\,{\bs
E}_{A_i}{\bs Q}^{m+1}\,.
\end{equation}

The parameter $\overline{w}=\overline{w}(q)$ satisfies the formula
\begin{equation}{\label{2.9}}
\overline{w}=w+\sum_{i=1}^ts_i\sum_{a\in
A_i}p_a=w+\sum_{i=1}^ts_i\| {{\bs E}_{A_i}p}\|_1\;.
\end{equation}

\end{remark}

\subsection{Equation for the leading eigenvalue $\overline{w}$ }

Here, using the notation and results from Section \ref{sec:4:1} we show that there exists an algebraic equation of degree at
most $t\cdot {\rm rk}_{\Z}R(X,d,\Gamma)$ for $\overline{w}$. Here
$R=R(X,d,\Gamma)$ is the orbital ring defined in Section~\ref{subs1.3}.

We can rewrite (\ref{2.6}), (\ref{2.7}) as follows
($\overline{w}$, defined in (\ref{2.9}), is considered to be a parameter here)
\begin{equation}{\label{2.12}} {\bs M}=\sum_{i=1}^t s_i{\bs E}_{A_i}{\bs L}\,,\quad {\bs
L}={\bs Q}(\overline{w}P_X(q){\bs I}-w{\bs Q})^{-1},
\end{equation}
\begin{equation}{\label{2.13}}
\sum_{i=1}^t s_i{\bs E}_{A_i}\bs{\hat p}=\sum_{i=1}^t s_i{\bs E}_{A_i}{\bs
L}\sum_{k=1}^t s_k{\bs E}_{A_k}\bs{\hat p},\quad \bs v=\sum_{i=1}^t s_i{\bs
E}_{A_i}\bs{\hat p}.
\end{equation}
Since ${\bs E}_{A_j}$ is a projection matrix, ${\bs
E}^2_{A_j}={\bs E}_{A_j}$, ${\bs E}_{A_j}{\bs E}_{A_i}={\bs 0}$
when $i\ne j$ and $s_j>0$, we can multiply both sides of
(\ref{2.13}) by ${\bs E}_{A_j}$. Then we obtain $t$ equalities
\begin{equation}{\label{2.14}}
{\bs E}_{A_j}\bs{\hat p}=\sum_{k=1}^t s_k{\bs E}_{A_j}{\bs L}{\bs
E}_{A_k}\bs{\hat p},\quad j=1,\dots, t.
\end{equation}

\begin{lemma}\label{lem:4:3} Let $\Gamma$-landscape be symmetric and $G$-invariant. Then not-trivial positive vector $\bs{\hat p}$ corresponding to the dominant eigenvalue $\overline w$ is constant on the
$G$-orbits $A_i$:
\begin{equation}{\label{2.15}}
({\bs E}_{A_i}\bs{\hat p})_x\equiv C_i>0\,,\;x\in A_i,\quad ({\bs
E}_{A_i}\bs{\hat p})_x\equiv 0 \,,\;x\notin A_i\,,\qquad i=1,\dots, t.
\end{equation}
\end{lemma}

\begin{proof} 

Take some solution $\bs z$ of problem \eqref{1.2}.
In fact, the
multiplication by ${\bs Q}$ commutes (see (\ref{1.13})) with
$G$-action $\bs{z}\mapsto g\bs{z}$, $g\bs{z}(x)=\bs{z}(g^{-1}x)$. Then ${\bs Q}{\bs{
w z}}=P_X(q)\overline{w}\bs{z}$ is equivalent to $g{\bs Q}g^{-1}
g{\bs w}g^{-1} g\bs{z}=P_X(q)\overline{w}\,g\bs{z}$, or, in view of
$G$-invariance, ${\bs Q}{\bs w}\,g\bs{z}=P_X(q)\overline{w}\,g\bs{z}$.
Hence, $g\bs{z}=g\bs{z}(x)$, $g\in G$, is also a solution. The averaging
$\bs{z}(x)\mapsto |G|^{-1}\sum g\bs{z}(x)$ provides a $G$-invariant
solution. In view of the Perron--Frobenius theorem the averaged solution  is proportional to $\bs z$ and moreover, is equal to $\bs z$ due to the last condition of \eqref{1.2}.
\end{proof}

The constants $C_i$ in \eqref{2.15} are to be normalized in such a way that
\begin{equation}{\label{2.16}}
\sum_{x\in A_0}\hat p_x+\sum_{i=1}^t C_i\,|A_i|=1\,.
\end{equation}
Then in view of (\ref{2.9}) we get
\begin{equation}{\label{2.17}}
\overline{w}=w+\sum_{i=1}^t s_i C_i\,|A_i|\;.
\end{equation}

\medskip In this case let $a\in A_j$ be a fixed point. The
equality ({\ref{2.14}) implies
\begin{equation}{\label{2.18}}
C_j=\sum_{k=1}^t s_k C_k \sum_{b\in A_k}l_{ab}\,,\qquad
(l_{ab})={\bs L}\,,\;\;j=1,\dots, t.
\end{equation}

\begin{proposition}\label{thm2.2} If $A_j$, $A_k$ are two $G$-orbits then
the inner sum $\sum_{b\in A_k}l_{ab}$ in \eqref{2.18} does not depend on the
choice of $a\in A_j$.
\end{proposition}

\begin{proof} Note that
$$
S_{a,A_k}=\sum_{b\in A_k}l_{ab}=\bs{ 1}^\top_a \bs L \bs{
1}_{A_k},
$$
where ${\bs 1}_a$ and ${\bs 1}_{A_k}$ are vector-columns
corresponding to the characteristic functions of the sets $\{a\}$
and $A_k$ respectively. For ${\bs L}$, given by (\ref{2.12}) and
considered as a kind of resolution (see (\ref{2.24}) below),  we
may assert  in view of (\ref{1.9}) that
$$
{\bs L}=\sum_{A\in {\rm Orb}}h_A(q) {\bs M}_A \;,
$$
$h_A(q)$ being some rational functions depending also on $w$,
$\overline{w}$ (see below the conjugate matrix ${{\bs T}^{-1}{\bs
L}{\bs T}}$). From (\ref{1.13}) we know that $\Gamma$- and,
consequently, $G$-action commute with the multiplication by each
${\bs M}_A$, hence, by ${\bs L}$. Then for $g\in G$
$$
S_{a,A_k}={\bs 1}^{\top}_a\bs L\bs{ 1}_{A_k}=({g^{-1}\bs
1}^\top)_a\bs L {g\bs 1}_{A_k}={\bs 1}^{\top}_{ga}\bs L
{\bs 1}_{A_k}=S_{ga,A_k},
$$
since ${\bs 1}_{A_k}$  corresponds to the characteristic function
of the set $A_k$, which is $G$-orbit. Since $ga$, $g\in G$, run
over $G$-orbit $A_j$ the proposition is proved.
\end{proof}

 In what follows we use notation $F_{jk}=S_{a,A_k}$ for any choice $a\in
A_j\,$. Let us conjugate ${\bs L}$ by the transition matrix ${\bs
T}$. In view of Theorem \ref{thm1.9} we obtain
\begin{align*}
F_{jk}&={\bs 1}^{\top}_a{\bs T}{\bs T}^{-1}{\bs L}{\bs T} {\bs
T}^{-1}{\bs 1}_{A_k}={\bs 1}^{\top}_a{\bs T}{\bs T}^{-1}{\bs
Q}{\bs T}(\overline{w}P_X(q){\bs I}-w{\bs T}^{-1}{\bs Q}{\bs
T}^{-1})^{-1}{\bs T}^{-1}{\bs 1}_{A_k}=\\
&={\bs 1}^{\top}_a{\bs T}\diag\left(\frac{P_x(q)}{\overline{w}P_X(q)-wP_x(q)}\right)
{\bs T}^{-1}{\bs
1}_{A_k}=\sum_{c=0}^{r-1}\frac{G^c_{jk}P_c(q)}{\overline{w}P_X(q)-wP_c(q)}\,,\quad
P_0(q)=P_X(q),
\end{align*}
for real algebraic numbers
\begin{equation}\label{2.19}
G^c_{jk}=\sum_{x: P_x(q)=P_c(q)}\sum_{b\in
A_k}t_{ax}t^{(-1)}_{xb}\,,\quad (t_{ax})={\bs
T},\;(t^{(-1)}_{xb})={\bs T}^{-1},\, a\in A_j,
\end{equation} since there are
at most $r={\rm rk}_{\Z}R=|{\rm Orb}|$ different eigenpolynomials
$P_x(q)$ in (\ref{1.14}).

Thus,
\begin{equation}\label{2.20}
F_{jk}=F_{jk}(q,w,\overline{w})= \sum_{c=0}^{r-1}\frac{G_{jk}^c
P_c(q)}{\overline{w}P_X(q)-wP_c(q)}\;.
\end{equation}
System  ({\ref{2.18}) now reads
\begin{equation}{\label{2.21}}
C_j=\sum_{k=1}^t F_{jk}s_k C_k \,,\qquad j=1,\dots, t.
\end{equation}
For the square matrix  ${\bs F}={\bs F}(\overline{w})=(F_{jk})$ of
order $t$, the positive diagonal matrix ${\bs S}=\diag(s_1,\ldots,s_t)$, and
the positive vector-column ${\bs c}=(C_k)$ we consequently have
\begin{equation}{\label{2.22}}
{\bs c}=\bs{FSc}.
\end{equation}

Summarizing the arguments in this section we thus have proved the following
\begin{theorem}\label{th:4:5}
Let $\Gamma$-landscape be symmetric and $G$-invariant. Then the dominant eigenvalue $\overline{w}$ of problem \eqref{1.2} satisfies the equation
\begin{equation}{\label{2.23}}
\det({\bs F}(\overline{w})-\bs{S}^{-1})=0\,.
\end{equation}
In view of (\ref{2.20}) this is an algebraic equation of degree at
most $t\cdot r=t\cdot {\rm rk}_\Z R(X,d,\Gamma)$ with coefficients
depending on $q$.
\end{theorem}

Now consider the simplest case when
$$
X=A_0\sqcup A_1,
$$
which we called \textit{two-valued fitness landscape} in \cite{semenov2016eigen}.

\begin{corollary}\label{cor2.4}
In conditions of Theorem \ref{th:4:5} let the fitness function
${\bs w}$ have $2$ values, $w$ on $A_0$ and $w+s$ on
$A_1=X\setminus A_0$, where $s>0$, $A_1$ is a single $G$-orbit.

Then the equation for $\overline{w}$ takes the form
\begin{equation}{\label{2.25}}
\sum_{c=0}^{r-1}\frac{G^c_{11}P_c(q)}{\overline{w}P_X(q)-wP_c(q)}=\frac{1}{s}\,,\quad
P_0(q)=P_X(q)\,.
\end{equation}
\end{corollary}

Finally, since the equation \eqref{2.23} in principle allows to find $\overline w$, we can use it to find the corresponding eigenvector $\bs{\hat p}$.
\begin{theorem}\label{thm2.3} Let a homogeneous $\Gamma$-landscape $(X,d,\Gamma, {\bs w})$
be symmetric and $G$-invariant, $G<\Gamma$. Suppose that  there is
a decomposition $X=A_0\sqcup\bigsqcup_{i=1}^t A_i\,$ such that
$A_0$ is a union of $G$-orbits on which ${\bs w}(A_0)\equiv w\geq
0$, and each $A_i$, $i=1,\dots, t$, is  a single $G$-orbit  on
which ${\bs w}(A_i)\equiv w+s_i$, where $s_i>0$. Suppose also that
in (\ref{2.19}) all coefficients $G^c_{jk}\geq 0$.

Then there exists a solution $\bs{\hat p}=\bs{\hat p}(q)$ of the
generalized Eigen's problem (\ref{1.2}) which is constant on
$G$-orbits.

\end{theorem}

\begin{proof} The (maximal) root $\overline{w}=\overline{w}(q)$  of
(\ref{2.23}) provides a non-trivial solution of (\ref{2.22}).
Consider the eigenvalue $\lambda=1$
 of  the matrix ${\bs F}{\bs S}$
with non-negative entries. It follows from the Perron--Frobenius
theorem that we can find a positive solution ${\bs c}=(C_k)$,
$k=1,\dots, t$ up to the positive scalar factor. Thus, we
can determine the projections ${\bs E}_{A_i}\bs{\hat p}=C_i{\bs
E}_{A_i}{\bs 1}$, where ${\bs 1}=(1,\dots,1)^\top$.

In view of (\ref{2.4}) solution $\bs{\hat p}$ of the problem (\ref{1.2})
can be reconstructed with the help of the formula
\begin{equation}{\label{2.24}}
\bs{\hat p}={\bs Q}(\overline{w}P_X(q){\bs I}-w{\bs Q})^{-1}\sum_{i=1}^t
s_i{\bs E}_{A_i}\bs{\hat p}=\sum_{i=1}^t s_i{\bs L}{\bs E}_{A_i}\bs{\hat p}.
\end{equation}

Lemma \ref{lem:4:3} implies that this solution $\bs{\hat p}$ is $G$-invariant.

The conditions (\ref{2.16}) and  (\ref{2.17}) (which is the same
as (\ref{1.3}))  enable us to determine the multiplication scalar
for $C_k$ and the final expression for $\bs{\hat p}$. The proof is complete.
\end{proof}

%
%

\begin{remark} For a two-valued symmetric and $G$-invariant
landscape $(X,d,\Gamma, {\bs w})$ satisfying the conditions of
Theorem \ref{thm2.3}, a solution  of (\ref{1.2}) can be essentially
simplified if $w=0$ (biologically, this is the case of {\it lethal} mutations). For
instance, see \cite[Example 4.8]{semenov2016eigen}.
\end{remark}

\section{Two examples: Polygonal and Hyperoctahedral landscapes}\label{sec:5}
In this section we show how the general theory of Sections \ref{sec:3} and \ref{sec:4} can be applied to some specific finite metric spaces $X$. Namely, we consider first the polygonal mutational landscape and then turn to analysis of the hyperoctahedral one. Two more detailed examples of the hypercube and regular simplex can be found in \cite{semenov2016eigen}. We would like to remark that although the classical Eigen's model is almost exclusively based on the geometry of binary cube, other mutational landscapes can be biologically relevant. For instance, in \cite{semenov2016eigen} we argued that the simplicial landscape is a natural description of the switching of the antigenic variants for some bacteria.
\subsection {Polygonal landscapes}\label{sec:5:1}

\subsubsection{Preliminaries} Let $X_l$ be the 0-skeleton of a
regular $l$-gon  with unit edges on a plane. We will assume that
$l\geq 3$, the case $l=2$ can be treated either directly, or as
the case of 1-dimensional simplex or the case of 1-dimensional
cube (i.e., a segment). Both cases were investigated in  \cite{semenov2016eigen}.

We will enumerate the points of $X_l$ by numbers of the set
$X_l=\{0,1,\dots,l-1\}$ with the fixed point 0 and the
counterclockwise enumeration of vertices (see Fig. \ref{fig:4}). It
is convenient to consider these numbers as elements of the cyclic
group $\Z/l\Z$, that is, consider the integer numbers modulo $l$.
Sometimes we will refer to the classical geometric interpretation
of $X_l$ as the set of all roots of unit of degree $l$ on the
complex plane $\C$, i.e.,
$X_l\cong\{1=\varepsilon^0,\varepsilon,\dots,\varepsilon^{l-1}\}$,
where $\varepsilon=e^{2\pi \I/l}$, so that $k\mod
l\leftrightarrow \varepsilon^k=e^{2\pi k \I/l}$.
\begin{figure}[!th]
\centering
\includegraphics[width=0.7\textwidth]{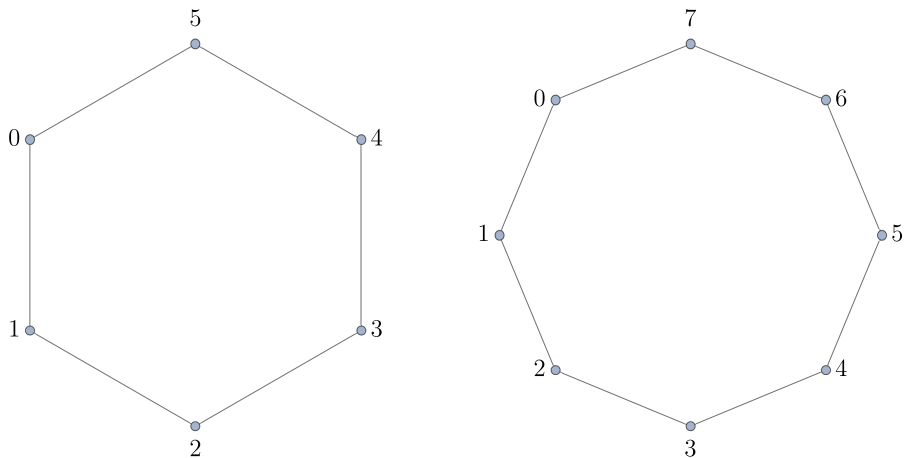}
\caption{Regular pentagon and hexagon}\label{fig:4}
\end{figure}

%
%
The metric $d$ is the so called {\it edge} metric on $X_l$:
the distance $d(k,j)$ between $k$ and $j$ is the minimal number of
edges of the regular $l$-gon connecting successively $k$ and $j$.
If $X_l\cong\Z/l\Z$ then $d(k,j)=\min\{|k-j|,l-|k-j|\}$.

For this metric if the cardinality $l=|X_l|=2N+1$ is odd then
$N=\diam X_l$ and there are two points, namely, $N$ and
$N+1$, such that $d(0,N)=d(0,N+1)=N$. If the cardinality
$l=|X_l|=2N$ is even then $N=\diam X_l$ and there is the
unique point $N$ for which $d(0,N)=N$. Moreover, we have  the
following distance polynomial:
\begin{equation}\label{3.1}
P_{X_l}(q)=\left\{ \begin{array}{ll}\vspace{3pt}
q^N+2\sum\limits_{k=1}^{N}(1-q)^kq^{N-k}\,,&l=2N+1\;\mbox{is odd},\\
q^N+2\sum\limits_{k=1}^{N-1}(1-q)^kq^{N-k}+(1-q)^N\,,&l=2N\;\mbox{is
even}\,.
 \end{array}
\right.
\end{equation}

It is well known that the group ${\rm Iso}(X_l)\cong D_l$ where
$D_l$ is a dihedral subgroup of order $2l$ which acts transitively
on $X_l$. The cyclic subgroup $C_l<D_l$ acts also transitively on
$X_l$ but the triple $(X_l,d,D_l)$ is symmetric in the sense of
Definition \ref{def:3:3} meanwhile the triple $(X_l,d,C_l)$ is not.

In what follows we consider only the symmetric polygonal
landscapes $(X_l,d,\Gamma, {\bs w})$, $\Gamma=D_l$. The stabilizer
$\Gamma_0={\rm St}_{\Gamma}(0)\cong \Z/2\Z$. If
$X_l\cong\{1=\varepsilon^0,\varepsilon,\dots,\varepsilon^{l-1}\}$,
then the unique non-trivial element of $\Gamma_0$ acts as the
complex conjugation. For the model $X_l\cong \Z/l\Z$ it acts by
the rule $k\mapsto l-k\mod l$.

There are exactly $N={\rm diam}(X_l)$ $\Gamma_0$-orbits in $X$:
$A_0=\{0\}$, $A_k=\{k,l-k\}$, $k=1,\dots,N-1$, and $A_N=\{N,N+1\}$
when $l=2N+1$ is odd, $A_N=\{N\}$ when $l=2N$ is even. In any case
each $A_k=S_k(0)$ is the sphere of radius $k$ centered at 0.

For the orbital ring $R_l=R(X_l,d,\Gamma)$ (see Section
\ref{subs1.3}) it means that $r={\rm rk}_\Z(X_l,d,\Gamma)=N+1$.
The orbital matrices ${\bs M}_k={\bs M}_{A_k}$ have the following
entries: $({\bs M}_k)_{ab}=1$ if $d(a,b)=k$ and $({\bs
M}_k)_{ab}=0$ otherwise (for a square $X_4$, see Example
\ref{ex1.8}).

The multiplication in the commutative ring $R_l=R(X_l,d,D_l)$ is
slightly different for the cases of odd and even $l$. In both
cases we have ${\bs M}_0{\bs M}_k={\bs M}_k$ since ${\bs M}_0={\bs
I}$ is the unity of $R_l$.
\begin{enumerate}
\item Case $l=2N+1$. It can be checked that
$${\bs M}^2_k=2{\bs M}_0+{\bs M}_{2k}\;\;\mbox {if}\; 0<2k\leq N,\quad{\bs
M}^2_k=2{\bs M}_0+{\bs M}_{l-2k} \;\;\mbox {if}\; N<2k\leq
2N=l-1\,.$$ Also we have for $0<k<j\leq N$:
$${\bs M}_k{\bs M}_j={\bs M}_{j-k}+{\bs M}_{j+k}\;\;\mbox {if}\; 0<j+k\leq N,
\quad{\bs M}_k{\bs M}_j={\bs M}_{j-k}+{\bs M}_{l-j-k} \;\;\mbox {if}\; N<j+k<
2N=l-1\,.$$

\item Case $l=2N$. We also check that
$${\bs M}^2_k=2{\bs M}_0+{\bs M}_{2k}\;\;\mbox {if}\; 0<2k< N,\quad{\bs
M}^2_k=2{\bs M}_0+{\bs M}_{l-2k} \;\;\mbox {if}\; N<2k< 2N=l\,.$$
If $2k=N$ then ${\bs M}^2_k=2{\bs M}_0+2{\bs M}_{N}$. If $k=N$
then ${\bs M}^2_N={\bs M}_0$.

\end{enumerate}
Also we have for $0<k<j< N$:
$${\bs M}_k{\bs M}_j={\bs M}_{j-k}+{\bs M}_{j+k}\;\;\mbox {if}\; 0<j+k< N,
\quad{\bs M}_k{\bs M}_j={\bs M}_{j-k}+{\bs M}_{l-j-k} \;\;\mbox
{if}\; N<j+k< 2N=l\,.$$ If $k+j=N$ we have
${\bs M}_k{\bs
M}_j={\bs M}_{j-k}+2{\bs M}_{N}$. If $j=N$ then ${\bs M}_k{\bs
M}_N={\bs M}_{N-k}$.

\begin{remark}It follows that the $\Z$-linear mapping $\rho\colon
R_l\longrightarrow\Z[2\cos(2\pi/l)]$ such that $\rho({\bs M}_0)=1$,
$\rho({\bs M}_k)=2\cos(2\pi k/l)$ ($k=1,\dots, N$ for the case
$l=2N+1$ and $k=1,\dots, N-1$ for the case $l=2N$) and, if $l=2N$,
$\rho({\bs M}_N)=-1$, is, in fact, a {\it ring homomorphism}. Note
also that $\Z[2\cos(2\pi/l)]=\R\cap \Z[\varepsilon]$ is the ring
of integers of the real field $\R\cap \Q[\varepsilon]$,
$Q[\varepsilon]$ being the cyclotomic field, since
$2\cos(2\pi/l)=\varepsilon+\varepsilon^{-1}$.
\end{remark}
\subsubsection{Transition matrix
${\bs T}={\bs T}_l$ and eigenpolynomials of the matrix ${\bs
Q}={\bs Q_l}$} Let the symmetric triple $(X_l,d,\Gamma)$,
$\Gamma=D_l$, be as in the previous subsection and let the columns
(rows) of all matrices under consideration be indexed by $0$,
\dots, $l-1$ modulo $l$ (since $X_l\cong \Z/l\Z$). Consider the
square symmetric matrix of order $l$
\begin{equation}\label{3.2}
{\bs T}={\bs T_l}=(t_{ab}):=\bigl(\cos(2\pi ab/l)-\sin(2\pi
ab/l)\bigr)\;, \quad a,b\in X_l\cong \Z/l\Z\,.
\end{equation}
\begin{theorem}\label{thm3.1}
 The matrix  ${\bs T}$
satisfies the following conditions:

1. The columns ${\bs t}_b$ of ${\bs T}$ compose a common
eigenbasis for all orbital matrices ${\bs M}_k$, $k=0,\dots, N$.

2. If $l=2N+1$ is odd then
$${\bs T}^{-1}{\bs M}_k{\bs
T}=2\diag(1,\cos(2\pi k/l), \cos(4\pi k/l),\dots, \cos(2\pi
(l-1)k/l)),\quad k=1,\dots, N\,.$$ If $l=2N$ is even then
$${\bs T}^{-1}{\bs M}_k{\bs
T}=2\diag(1,\cos(2\pi k/l), \cos(4\pi k/l),\dots,
\cos(2\pi(l-1)k/l)),\quad k=1,\dots, N-1\,,$$
$${\bs T}^{-1}{\bs M}_N{\bs
T}=\diag(1,-1, 1,-1,\dots,1,-1)\,.$$

3. ${\bs T}^{-1}{\bs Q}{\bs T}=\diag(P_0(q),P_1(q),\dots,
P_{l-1}(q))$ where $P_0(q)=P_{X_l}(q)$ is the distance polynomial
(\ref{3.1}).

If $l=2N+1$ is odd then
\begin{equation}\label{3.3}
p_j(q)=p_{l-j}(q)=q^N+\sum_{k=1}^N 2\cos(2\pi
kj/l)\,(1-q)^kq^{N-k}\;,\quad j=1,\dots, N \;.
\end{equation} If $l=2N$ is even then
\begin{equation}\label{3.4}
p_j(q)=p_{l-j}(q)=q^N+\sum_{k=1}^{N-1} 2\cos(2\pi
kj/l)\,(1-q)^kq^{N-k}+(-1)^j(1-q)^N\;,\quad j=1,\dots, N\;.
\end{equation}

4. ${\bs T}^2=l\,{\bs I}$. In other words, ${\bs
T}^{-1}=\frac{1}{l}{\bs T}$.
\end{theorem}

\begin{proof} 1 and 2. Consider cyclic matrices ${\bs C}_k$
with only $l$ non-trivial entries $({\bs C}_k)_{a,a+k}=1$ where
subindices are taken modulo $l$. Note that ${\bs C}_0={\bs
M}_0={\bs I}$ in any case and ${\bs C}_N={\bs M}_N$ if $l=2N$ is
even. In the other cases ${\bs C}_k+{\bs C}_{l-k}={\bs M}_k$.

Consider also the vectors ${\bs
v}_j=(1,\varepsilon^j,\varepsilon^{2j},\dots,\varepsilon^{(l-1)j})^T$,
$j\in X_l$, $\varepsilon=e^{2\pi \I/l}$. Straightforward checking
yields
$$
{\bs C}_k{\bs v}_j=\varepsilon^{kj}{\bs v}_j\;.
$$
Note that the Vandermonde determinant $\det(\varepsilon^{kj})\ne
0$. It follows that
$$
{\bs C}_k{\bs v}_j=\varepsilon^{kj}{\bs v}_j\;,\quad {\bs
C}_{l-k}{\bs v}_{j}=\varepsilon^{-kj}{\bs v}_{j}\;,\quad{\bs
C}_k{\bs v}_{l-j}=\varepsilon^{-kj}{\bs v}_{l-j}\;,\quad {\bs
C}_{l-k}{\bs v}_{l-j}=\varepsilon^{kj}{\bs v}_{l-j}\,
$$
since $\varepsilon^l=1$. If $l=2N+1$ is odd then the real
vector-columns ${\bs t}_0={\bs v}_0$, ${\bs t}_j={\rm Re}\,{\bs
v}_j-{\rm Im}\,{\bs v}_{l-j}$, ${\bs t}_{l-j}={\rm Re}\,{\bs
v}_{l-j}+{\rm Im}\,{\bs v}_{j}$ are linearly independent
eigenvectors of each ${\bs M}_k$ such that
$$
{\bs M}_0 {\bs t}_j={\bs t}_j\,,\quad {\bs M}_k {\bs
t}_j=(\varepsilon^{kj}+\varepsilon^{-kj}){\bs t}_j=2\cos(2\pi
kj/l){\bs t}_j,\quad k=1,\dots,N\,.
$$
If $l=2N$ is even then the  vector-columns ${\bs t}_0={\bs v}_0$ ,
${\bs t}_N=(1,-1, 1,-1,\dots,1,-1)^\top$, ${\bs t}_j={\rm Re}\,{\bs
v}_j-{\rm Im}\,{\bs v}_{l-j}$,
 ${\bs t}_{l-j}={\rm
Re}\,{\bs v}_{l-j}+{\rm Im}\,{\bs v}_{j}$ are linearly independent
eigenvectors of each ${\bs M}_k$ such that
$$
{\bs M}_0 {\bs t}_j={\bs t}_j\,,\quad {\bs M}_k {\bs
t}_j=\varepsilon^{kj}+\varepsilon^{-kj}{\bs t}_j=2\cos(2\pi
kj/l){\bs t}_j,\quad k=1,\dots,N-1\,,\quad {\bs M}_N {\bs
t}_j=\cos(\pi j){\bs t}_j\,.$$

In any case the transition matrix ${\bs T}$ has the form
(\ref{3.2}). This finishes the proof of the assertions 1 and 2.

3. Recall (\ref{1.9}) that ${\bs Q}= \sum_{k=0}^N (1-q)^k
q^{N-k}{\bs M}_k\,$. In view of the assertion 2  $${\bs
T}^{-1}{\bs Q}{\bs T}=\sum_{k=0}^N (1-q)^k q^{N-k}{\bs T}^{-1}{\bs
M}_k{\bs T}=\diag(p_0(q),p_1(q),\dots, p_{l-1}(q))\,.$$ Comparing
the diagonal entries we get the desired result.

4.  Straightforward calculations with trigonometric sums. The
theorem is proved.
\end{proof}

\begin{remark} Note that Conjecture \ref{conj1.11} is true for the
triple $(X_l,d,\Gamma)$.
\end{remark}

%
\subsubsection{On single peaked and alternating $G$-invariant polygonal landscapes}
It is known that the subgroups of the dihedral group $D_l$ are, up
to isomorphism, the following groups: dihedral groups $D_m$ and
cyclic groups $C_m$ for $m$ dividing $l$. Note that $D_1\cong
\Z/2\Z$.

In this section the explicit expression of the equation
(\ref{2.25}) is given for two-valued fitness landscapes
$(X_l,d,\Gamma=D_l,{\bs w})$. Specifically, we consider
\begin{enumerate}
\item Single peaked landscapes when $X=A_0\sqcup
A_1$ and $A_1$ consists of a single point, say, $A_1=\{0\}$. Thus,
${\bs w}(0)=w+s$, ${\bs w}(x)=w$, $x\ne 0$. This landscape is
$G$-invariant for $G=\{1\}$.

\item Alternating landscapes for $l$-gons with even $l=2N$ (see Fig. \ref{fig:5}). For
these landscapes $X=A_0\sqcup A_1$ where $A_1=\{0,2,4,\dots,l-2\}$
and ${\bs w}(A_0)=w$, ${\bs w}(A_1)=w+s$. This landscape is
$G$-invariant for the cyclic group $G=C_N$.
\end{enumerate}
\begin{figure}[!th]
\centering
\includegraphics[width=0.7\textwidth]{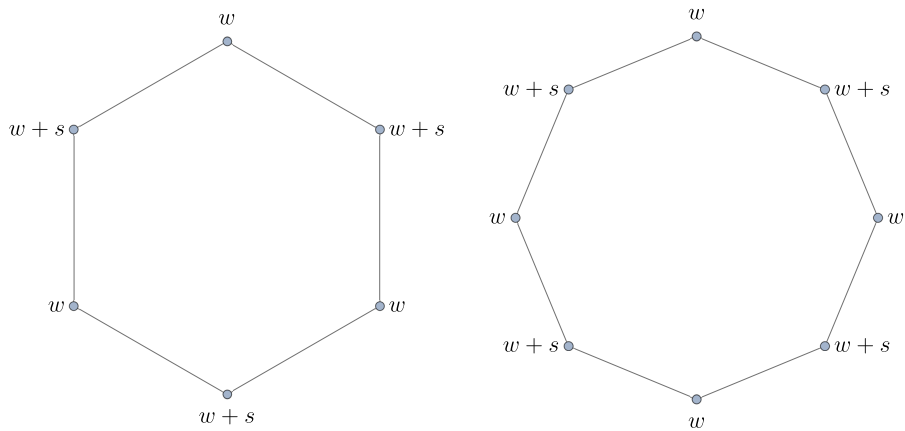}
\caption{Alternating hexagonal and octagonal landscapes}\label{fig:5}
\end{figure}
%
%
%
%

1. In the case of a single peaked landscape the equation
(\ref{2.25}) reads
\begin{equation}\label{3.5}
\frac{1}{\overline{w}-w}+ \sum_{c=1}^{N-1}\frac{2P_c(q)}
{\overline{w}P_X(q)-wP_c(q)}+\frac{HP_N(q)}
{\overline{w}P_X(q)-wP_N(q)}=\frac{l}{s}\,,\qquad P_0(q)=P_X(q)\,,
\end{equation}
where $H=2$ if $l=2N+1$ and $H=1$ if $l=2N$. The polynomials
$P_c(q)$ are defined in (\ref{3.1}), (\ref{3.3}), (\ref{3.4}).

Indeed, we need to calculate the numbers $G^{c}_{11}$ with the
help of (\ref{2.19}). But $a=0$, $A_1=\{0\}$ and in view of
Theorem \ref{3.2} $t_{0x}=lt^{(-1)}_{x0}=1$ for the transition
matrix ${\bs T}$. Hence the result.
\begin{figure}[!th]
\centering
\includegraphics[width=0.4\textwidth]{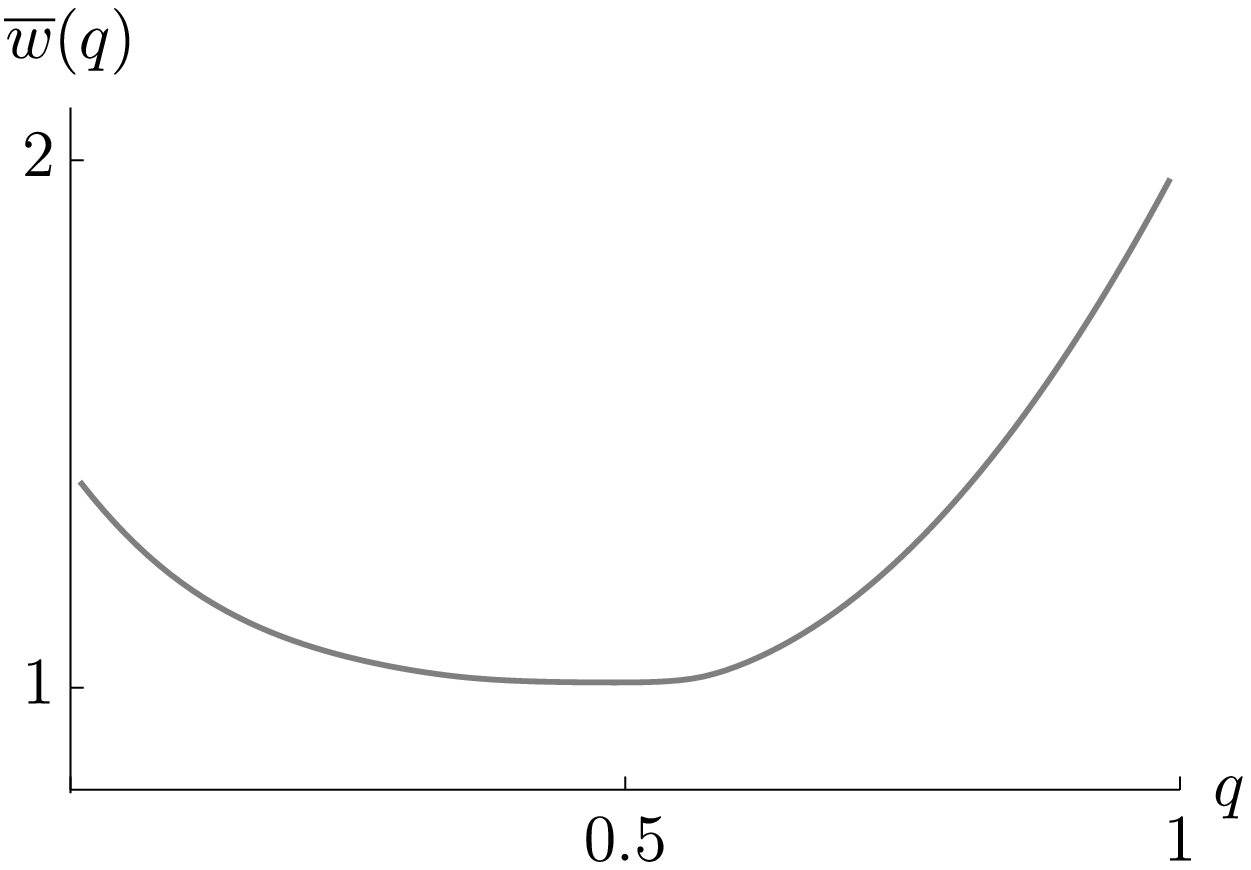}$\qquad$
\includegraphics[width=0.4\textwidth]{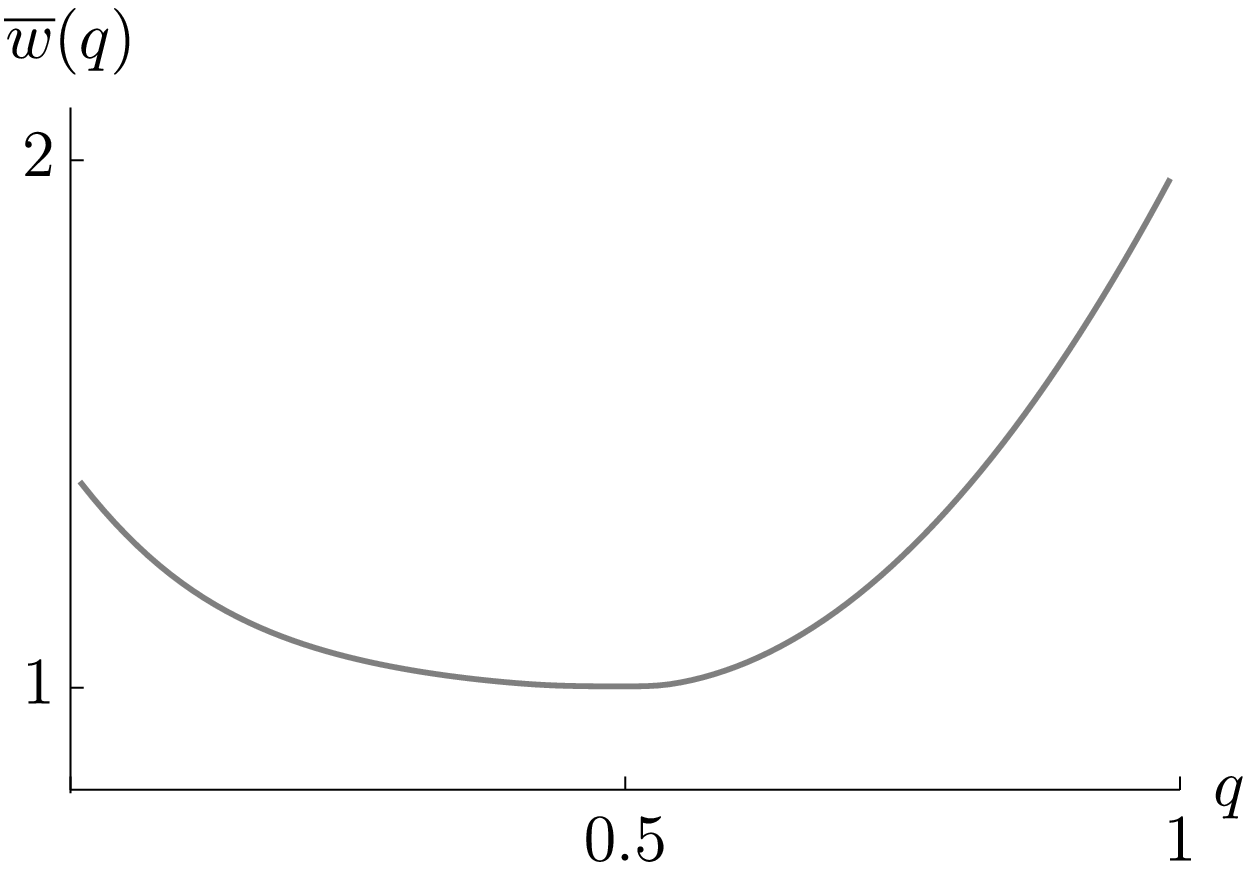}
\caption{The leading eigenvalue of spectral problem \eqref{1.2} on regular $2N$-gon for single peaked landscape $w=1,s=1$. On the left $N=50$, on the right $N=200$.}\label{fig:6}
\end{figure}

In Fig. \ref{fig:6} we present two numerical examples for the considered situation. Note the absence of non-analytical behavior of $\overline w$, i.e., the absence of the error threshold (or phase transition).


2.  In the case of an alternating landscape with $l=2N$ the
straightforward calculation yields $G^0_{11}=G^N_{11}=1/2$,
$G^k_{11}=0$ for $0<k<N$. Then the equation (\ref{2.25}) is
quadratic:
$$
\frac{1}{\overline{w}-w}+ \frac{P_N(q)}
{\overline{w}P_X(q)-wP_N(q)}=\frac{2}{s}\,,\qquad P_0(q)=P_X(q)\,,
$$ or, for $u=w/s$, $\overline{u}=\overline{w}/s$,
\begin{equation}\label{3.7}
\frac{1}{\overline{u}-u}+ \frac{P_N(q)}
{\overline{u}P_0(q)-uP_N(q)}=2\,.
\end{equation}
 The polynomials $P_0(q),\, P_N(q)\in\Z[q]$:
 $$
P_0(q)=q^N+2\sum\limits_{k=1}^{N-1}(1-q)^kq^{N-k}+(1-q)^N\,,\quad
P_N(q)=q^N+2\sum\limits_{k=1}^{N-1}(-1)^k(1-q)^kq^{N-k}+(-1)^N(1-q)^N\,.$$

The solution of (\ref{3.7}) is (the positive square root is taken)
\begin{equation}\label{3.8}
\overline{u}(q)=\frac{(2u+1)(P_0(q)+P_N(q))+\sqrt{(2u+1)^2(P_0(q)+P_N(q))^2-16u(u+1)P_0(q)P_N(q)}}
{4P_0(q)}\,,
\end{equation}
numerical examples are given in Fig. \ref{fig:7}.
\begin{figure}[!th]
\centering
\includegraphics[width=0.4\textwidth]{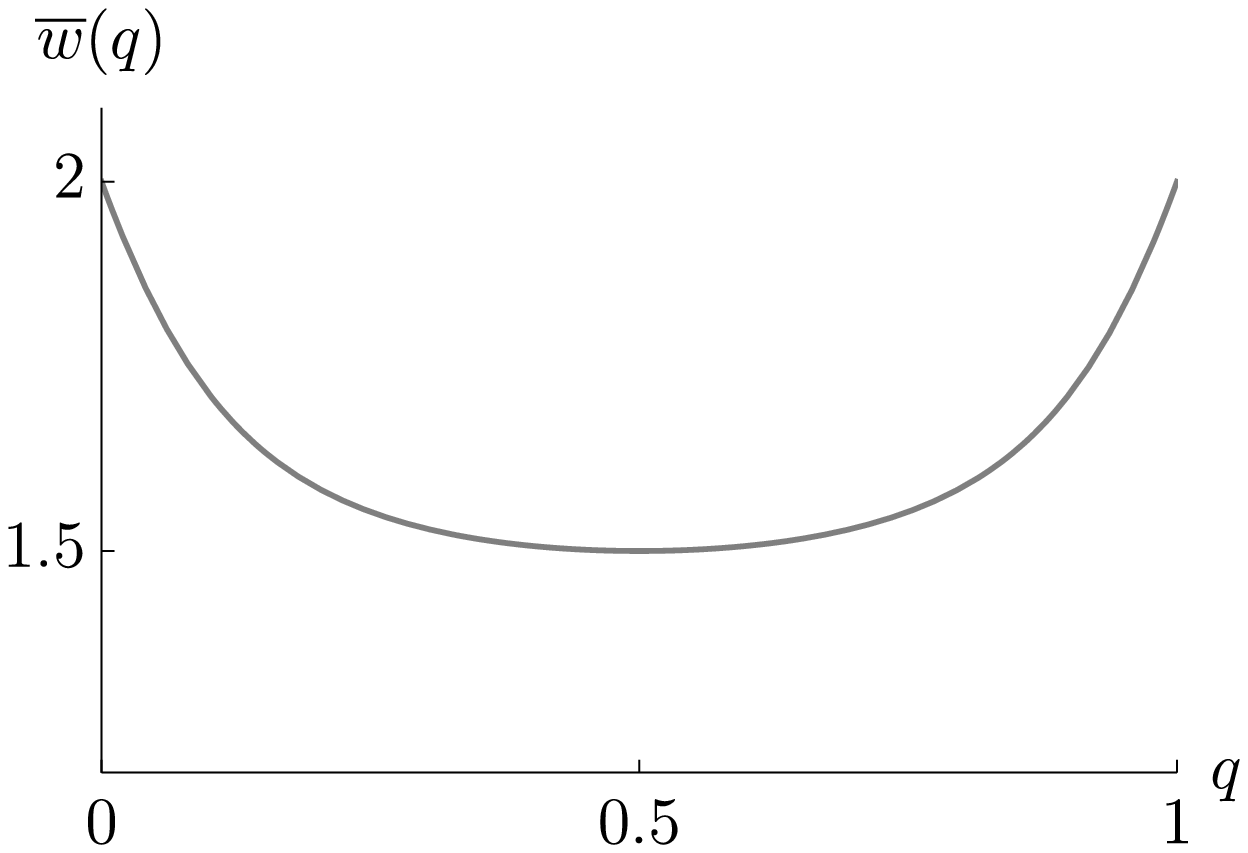}$\qquad$
\includegraphics[width=0.4\textwidth]{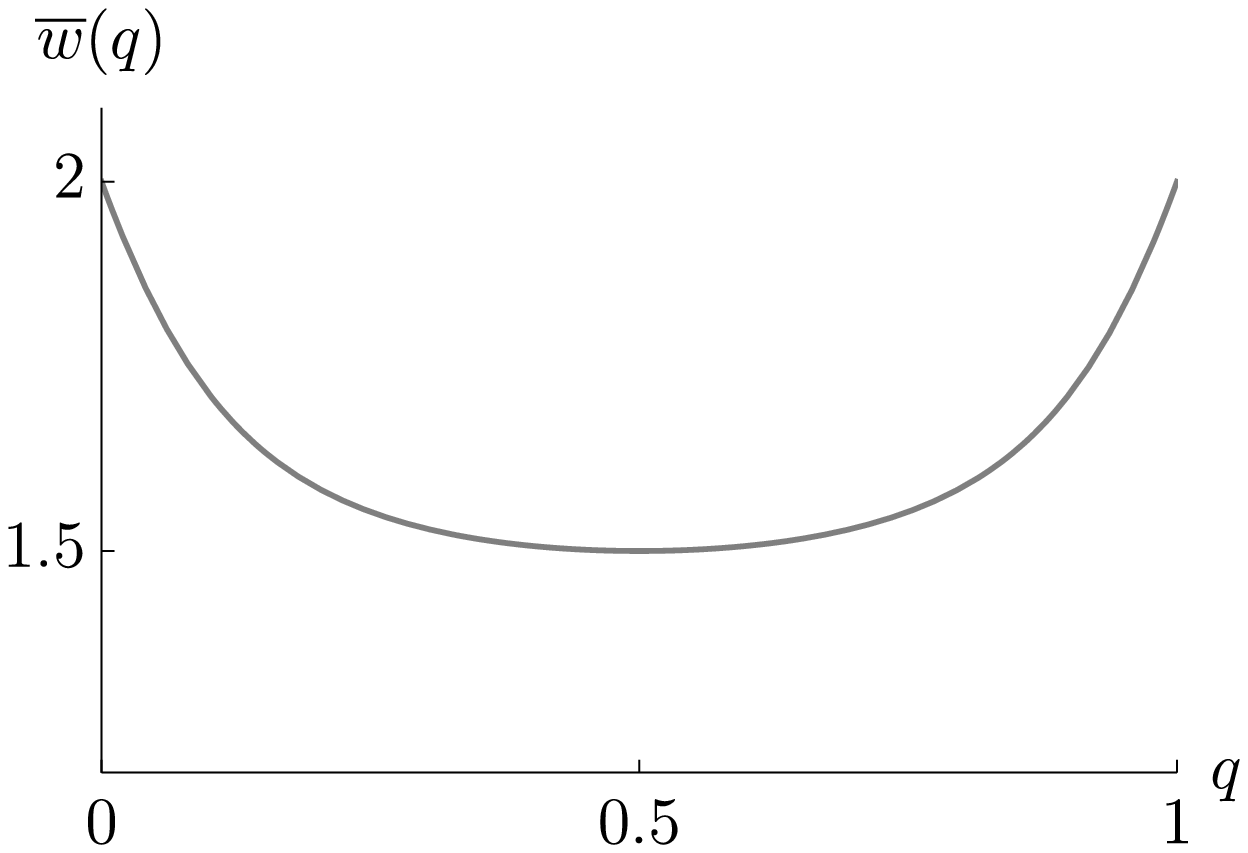}
\caption{The leading eigenvalue of spectral problem \eqref{1.2} on regular $2N$-gon for alternating fitness ladscape landscape $w=1,s=1$. On the left $N=50$, on the right $N=200$.}\label{fig:7}
\end{figure}

\subsection {Hyperoctahedral or dual Eigen's landscapes}

\subsubsection{Preliminaries} Let $X_n=\{x_0,\dots,x_k,\dots,x_{2n-1-k},\dots,x_{2n-1}\}$
be the 0-skeleton of a regular $n$-dimensional  hyperoctahedron
which is the convex hull of the vertices (in
$\R^n=\{(\xi_1,\dots,\xi_n)\}$)
$$
x_0=(1,0,\dots,0) ,\,x_{2n-1}=(-1,0,\dots,0),\;\dots,$$
$$x_k=(0,\dots,0,1,0,\dots,0) \;(\mbox{1
stands on $(k+1)$-th position}), \quad
x_{2n-1-k}=(0,\dots,0,-1,0,\dots,0),\;\dots\;.
$$

A classical octahedron ($n=3$) is represented in
Figure \ref{fig:8}.

\begin{figure}[!th]
\centering
\includegraphics[width=0.45\textwidth]{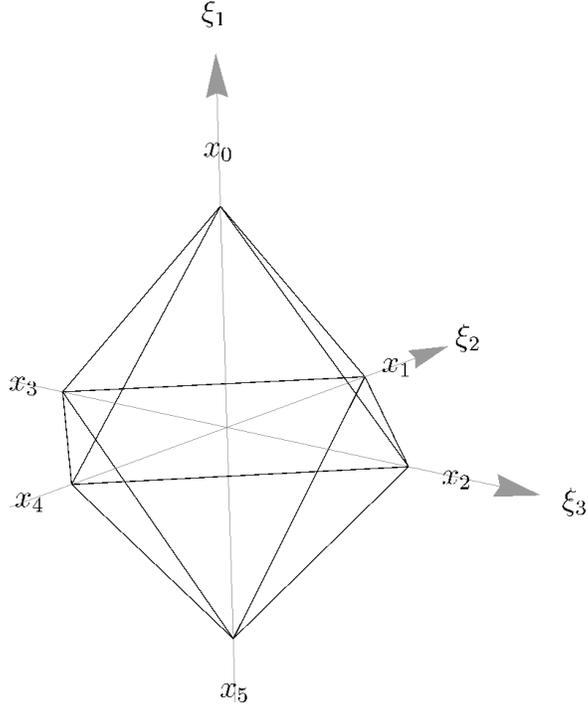}
\caption{Regular octahedron.}\label{fig:8}
\end{figure}
%
%
%
%
%
%
%
%
%

The metric $d$ is again the {\it edge} metric metric on $X_n$:
the distance $d(x_k,x_j)$ is defined as follows
$$
d(x_k,x_j)=\left\{
\begin{array}{l}
0,\;\mbox{if}\quad k=j,\\
2, \;\mbox{if}\quad k+j=2n-1,\\
1, \;\mbox{otherwise}\,.
\end{array}
\right.
$$

We have the cardinality $l=|X_n|=2n$ and  ${\rm diam}(X_n)=N=2$.
The distance polynomial is $P_{X_n}(q)=q^2+(2n-2)(1-q)q+(1-q)^2$.

Since a hyperoctahedron is the dual polytope to the hypercube
appearing in the classical Eigen's model, the isometry group
$\Gamma=\Gamma_n={\rm Iso} (X_n)$ is a hyperoctahedral group.
$\Gamma_n$ is isomorphic as an abstract group to the Weyl group of
the root system of type $B_n$ or $C_n$ and $|\Gamma_n|=2^n\cdot
n!$.

Note that the triple $(X_n,d,\Gamma_n)$ is symmetric  in the sense
of definition \ref{def:3:3} and that the stabilizer $\Gamma_0={\rm
St}_{\Gamma_n}(x_0)\cong {\rm Iso} (X_{n-1})$. Indeed, viewing
$x_0$ and $x_{2n-1}$ as the "north" and "south" poles
respectively, we have an isometry action of $\Gamma_0={\rm
St}_{\Gamma}(x_0)$ on the "equatorial" hyperoctahedron $X_{n-1}$
which is 1-sphere of $x_0$ with respect to the metric $d$. Then it
is not hard to see that $\Gamma_0={\rm St}_{\Gamma_n}(x_0)$ is the
group of all isometries of $X_{n-1}$.

If $n\geq 2$ then there are exactly $3={\rm diam}(X_n)+1$
$\Gamma_0$-orbits in $X_n$: namely, $A_0=\{x_0\}$,
$A_1=\{x_1,\dots,x_{2n-2}\}\cong X_{n-1}$, and $A_2=\{x_{2n-1}\}$.

For the orbital ring $R_n=R(X_n,d,\Gamma)$ (see Section
\ref{subs1.3}) it means that $r={\rm rk}_\Z(X_l,d,\Gamma)=3$. The
orbital matrices ${\bs M}_k={\bs M}_{A_k}$ have the following
entries: $({\bs M}_k)_{ab}=1$ if $d(a,b)=k$ and $({\bs
M}_k)_{ab}=0$ otherwise. For a 2-octahedron (a square) $X_2$ see
Example \ref{ex1.8}, for a 3-octahedron $X_3$ (a classical one)
$$
{\bs M_0}={\bs I}=\left[
\begin{array}{cccccc}
1&0&0&0&0&0\\
0&1&0&0&0&0\\
0&0&1&0&0&0\\
0&0&0&1&0&0\\
0&0&0&0&1&0\\
0&0&0&0&0&1\\
\end{array}\right],\, {\bs M_1}=\left[
\begin{array}{cccccc}
0&1&1&1&1&0\\
1&0&1&1&0&1\\
1&1&0&0&1&1\\
1&1&0&0&1&1\\
1&0&1&1&0&1\\
0&1&1&1&1&0\\
\end{array}\right],\,
{\bs M_2}=\left[
\begin{array}{cccccc}
0&0&0&0&0&1\\
0&0&0&0&1&0\\
0&0&0&1&0&0\\
0&0&1&0&0&0\\
0&1&0&0&0&0\\
1&0&0&0&0&0
\end{array}\right].
$$

The multiplication table of these matrices in
$R_n=R(X_n,d,\Gamma_n)$ is as follows ($n\geq 2$):
$$
\begin{array}{|c||c|c|c|}\hline
\times&{\bs I}&{\bs M_1}&{\bs M_2}\\ \hline {\bs I}&{\bs I}&{\bs
M_1}&{\bs M_2}\\ \hline {\bs M_1}&{\bs M_1}&(2n-2)({\bs I}+{\bs
M_2})+(2n-4){\bs M_1}&{\bs M_1}\\\hline {\bs M_2}&{\bs M_2}&{\bs
M_1}&{\bs I}\\\hline
\end{array}
$$

\subsubsection{Transition matrix
${\bs T}={\bs T}_n$ and eigenpolynomials of the matrix ${\bs
Q}={\bs Q_n}$} Let the symmetric triple $(X_n,d,\Gamma_n)$
 be as in the previous subsection and let the columns
(rows) of all matrices under consideration be indexed by $0$,
\dots, $2n-1$ corresponding to $x_0,\dots, x_{2n-1}$. Consider the
following four matrices of order $n$ ($n\geq 2$):
$$
{\bs T_{00}}=\left[
\begin{array}{crrcr}
1&1&1&\dots&1\\
1&-1&0&\dots&0\\
1&0&-1&\dots&0\\
\vdots&\vdots&\vdots&\ddots&\vdots\\
1&0&0&\cdots&-1\\
\end{array}\right],\;
{\bs T_{01}}=\left[
\begin{array}{rlrrr}
1&\dots&1&1&1\\
0&\dots&0&-1&1\\
0&\dots&-1&0&1\\
\vdots&.^{\,\displaystyle{.^{\,.}}}&\vdots&\vdots&\vdots\\
-1&\cdots&0&0&1\\
\end{array}\right],
$$
$$
{\bs T_{10}}=\left[
\begin{array}{crrcr}
1&0&0&\cdots&-1\\
\vdots&\vdots&\vdots&.^{\,\displaystyle{.^{\,.}}}&\vdots\\
1&0&-1&\dots&0\\
1&-1&0&\dots&0\\
1&1&1&\dots&1\\
\end{array}\right],\;
{\bs T_{11}}=\left[
\begin{array}{rrrrr}
1&\cdots&0&0&-1\\
\vdots&\ddots&\vdots&\vdots&\vdots\\
0&\dots&1&0&-1\\
0&\dots&0&1&-1\\
-1&\dots&-1&-1&-1\\
\end{array}\right].
$$
Here the entries of the first row and column of the matrix ${\bs
T_{00}}$ are equal to 1, the diagonal entries, except for the
first one, are equal to $-1$, and the other entries are trivial.
The matrices ${\bs T_{01}}$ and ${\bs T_{10}}$ are the horizontal
and vertical mirror copies of ${\bs T_{00}}$, the matrix ${\bs
T_{11}}$ is the horizontal mirror copy of ${\bs T_{10}}$
multiplied by $-1$.

 Consider the square symmetric matrix of order $2n$
\begin{equation}\label{4.1}
{\bs T}={\bs T_n}=\left[
\begin{array}{cc}
{\bs T_{00}}&{\bs
T_{01}}\\
{\bs T_{10}}&{\bs
T_{11}}\\
\end{array}\right]\,.
\end{equation}

We also introduce the following four matrices of order $n$:
$$
{\bs K_{00}}=\left[
\begin{array}{crrcr}
1&1&1&\dots&1\\
1&1-n&1&\dots&1\\
1&1&1-n&\dots&1\\
\vdots&\vdots&\vdots&\ddots&\vdots\\
1&1&1&\cdots&1-n\\
\end{array}\right],\;
{\bs K_{01}}=\left[
\begin{array}{rlrrr}
1&\dots&1&1&1\\
1&\dots&1&1-n&1\\
1&\dots&1-n&1&1\\
\vdots&.^{\,\displaystyle{.^{\,.}}}&\vdots&\vdots&\vdots\\
1-n&\cdots&1&1&1\\
\end{array}\right],
$$
$$
{\bs K_{10}}=\left[
\begin{array}{crrcr}
1&1&1&\cdots&1-n\\
\vdots&\vdots&\vdots&.^{\,\displaystyle{.^{\,.}}}&\vdots\\
1&1&1-n&\dots&1\\
1&1-n&1&\dots&1\\
1&1&1&\dots&1\\
\end{array}\right],\;
{\bs K_{11}}=\left[
\begin{array}{rrrrr}
n-1&\cdots&-1&-1&-1\\
\vdots&\ddots&\vdots&\vdots&\vdots\\
-1&\dots&n-1&-1&-1\\
-1&\dots&-1&n-1&-1\\
-1&\dots&-1&-1&-1\\
\end{array}\right].
$$
Here the diagonal entries of the matrix ${\bs K_{00}}$, except for
$k_{00}=1$, are equal to $1-n$, and the other entries are equal to
1. The matrices ${\bs K_{01}}$ and ${\bs K_{10}}$ are the
horizontal and vertical mirror copies of ${\bs K_{00}}$, the
matrix ${\bs K_{11}}$ is the horizontal mirror copy of ${\bs
K_{10}}$ multiplied by $-1$.

 Consider the square symmetric matrix of order $2n$
\begin{equation}\label{4.2}
{\bs K}={\bs K_n}=\left[
\begin{array}{cc}
{\bs K_{00}}&{\bs
K_{01}}\\
{\bs K_{10}}&{\bs
K_{11}}\\
\end{array}\right]\,.
\end{equation}

\begin{theorem}\label{thm4.1}
 The matrix  ${\bs T}={\bs T}_n$, $n\geq 2$,
satisfies the following conditions:

1. $\det({\bs T}_n)=(-1)^n\cdot 2^n\cdot n^2$, consequently,
${\bs T}_n$ is a non-degenerate matrix.

2. ${\bs T}^{-1}_n=\frac{1}{2n}{\bs K}_n$.

3. The columns ${\bs t}_b$ ($b=0,\dots,2n-1$) of ${\bs T}$ compose a
common eigenbasis for each orbital matrix ${\bs M}_0$, ${\bs M}_1$, ${\bs M}_2$.

4. More precisely,
$${\bs T}^{-1}{\bs M}_1{\bs
T}=\diag(2n-2,\underbrace{-2,\dots,
-2}_{n-1\;\mbox{\scriptsize{times}}},\underbrace{0,\dots,
0}_{n\;\mbox{\scriptsize{times}}})\,,$$
$${\bs T}^{-1}{\bs M}_2{\bs
T}=\diag(\underbrace{1,\dots,
1}_{n\;\mbox{\scriptsize{times}}},\underbrace{-1,\dots,
-1}_{n\;\mbox{\scriptsize{times}}})\,.$$

5. Let
${\bs Q}={\bs Q}_n=\left((1-q)^{d(x_a,x_b)}\cdot
q^{2-d(x_a,x_b)}\right)$. Then
  $${\bs T}^{-1}{\bs Q}{\bs
T} =\diag(P_0(q),P_1(q),\dots, P_{2n-1}(q))=$$
$$=\diag(q^2+(2n-2)(1-q)q+(1-q)^2,\underbrace{(2q-1)^2,\dots,
(2q-1)^2}_{n-1\;\mbox{\scriptsize{times}}},\underbrace{2q-1,\dots,
2q-1}_{n\;\mbox{\scriptsize{times}}})\,.
$$
\end{theorem}

\begin{proof} 1. Subtracting the row 0 from the row $2n-1$,
the row 1 from the row $2n-2$, \dots ,  the row $n-1$ from the row
$n$ (note that the subindices of matrix entries range from $0$ to
$2n-1$) we get
$$
\det({\bs T})=\det\left[
\begin{array}{cc}
{\bs T_{00}}&{\bs
T_{01}}\\
{\bs T_{10}}&{\bs
T_{11}}\\
\end{array}\right]=\det\left[
\begin{array}{cc}
{\bs T_{00}}&{\bs
T_{01}}\\
{\bs 0}&{\bs
2T_{11}}\\
\end{array}\right]=2^n\det({\bs T_{00}})\det({\bs
T_{11}}).
$$

Adding the sum of columns 2, \dots, $n$ to the first column of the
matrix ${\bs T_{00}}$ we obtain the equality $\det({\bs
T_{00}})=(-1)^{n-1}n$. In a similar way we obtain that $\det({\bs
T_{11}})=-n$. Hence the desired result.

2.  Straightforward checking shows that ${\bs T}_n{\bs
K_n}=2n{\bs I}$. Note that the matrix ${\bs T_{00}}$ appears as a
transition matrix for a simplicial landscape, see \cite[Section
6]{semenov2016eigen} for the inverse ${\bs T^{-1}_{00}}$ and for more details.

3 and 4. Straightforward calculations show that
$$
{\bs M_{1}}{\bs t}_0=(2n-2) {\bs t}_0,\quad {\bs M_{1}}{\bs
t}_b=-2 {\bs t}_b\;\;(b=1,\dots, n-1), \quad {\bs M_{1}}{\bs
t}_b={\bs 0}\;\;(b=n,\dots, 2n-1)\,,
$$
$$
 {\bs M_{2}}{\bs
t}_b= {\bs t}_b\;\;(b=0,\dots, n-1), \quad {\bs M_{1}}{\bs
t}_b=-{\bs t}_b\;\;(b=n,\dots, 2n-1)\,.
$$

5. Since ${\bs Q}=q^2{\bs I}+q(1-q){\bs M}_1+(1-q)^2{\bs
M}_2$ and in view of 3 we get
$${\bs T}^{-1}{\bs Q}{\bs
T}=q^2{\bs I}+q(1-q)\diag(2n-2,\underbrace{-2,\dots,
-2}_{n-1\;\mbox{\scriptsize{times}}},\underbrace{0,\dots,
0}_{n\;\mbox{\scriptsize{times}}})+(1-q)^2\diag(\underbrace{1,\dots,
1}_{n\;\mbox{\scriptsize{times}}},\underbrace{-1,\dots,
-1}_{n\;\mbox{\scriptsize{times}}})=
$$
$$
=\diag(q^2+(2n-2)(1-q)q+(1-q)^2,\underbrace{(2q-1)^2,\dots,
(2q-1)^2}_{n-1\;\mbox{\scriptsize{times}}},\underbrace{2q-1,\dots,
2q-1}_{n\;\mbox{\scriptsize{times}}})\,.
$$
The theorem is proved.
\end{proof}

\begin{remark} Note that Conjecture \ref{conj1.11} is true for the
triple $(X_n,d,\Gamma_n)$.
\end{remark}

\subsubsection{On single peaked  $G$-invariant hyperoctahedral landscapes}

In this subsection the explicit expression of the equation
(\ref{2.25}) is given for 2-fitness hyperoctahedral landscapes
$(X_n,d,\Gamma_n,{\bs w})$.

 Consider a single peaked landscape
 for which $X_n=A_0\sqcup A_1$ and $A_1$ consists of a
single point, say, $A_1=\{x_0\}$. Thus, ${\bs w}(x_0)=w+s$, ${\bs
w}(x_k)=w$, $k=1,\dots,2n-1$. This landscape is
$G$-invariant under the action of the trivial group $G=\{1\}$.

In the case of a single peaked landscape the equation (\ref{2.25})
reads
\begin{equation}\label{4.3}
\frac{G_{11}^0}{\overline{w}-w}+\frac{G_{11}^1(2q-1)^2}
{\overline{w}P_X(q)-w(2q-1)^2}+\frac{G_{11}^{2}(2q-1)}
{\overline{w}P_X(q)-w(2q-1)}=\frac{l}{s}\,,
\end{equation}
where $P_X(q)=p_0(q)=q^2+(2n-2)q(1-q)+(1-q)^2$ is the distance
polynomial and, in view of the assertion 2 of Theorem \ref{thm4.1}
and (\ref{2.19}),
$$
G_{11}^0=\frac{1}{2n}\,\quad G_{11}^1=\frac{n-1}{2n},\quad
G_{11}^2=\frac{n}{2n}=\frac{1}{2}\,.
$$

Finally, for the parameters $u=w/s$, $\overline{u}=\overline{w}/s$
we obtain the (cubic) equation
\begin{equation}\label{4.4}
\frac{1}{\overline{u}-u}+\frac{(n-1)(2q-1)^2}
{\overline{u}P_X(q)-u(2q-1)^2}+\frac{n(2q-1)}
{\overline{u}P_X(q)-u(2q-1)}=2n\,.
\end{equation}
Numerical illustrations are given in Fig. \ref{fig:9}. Note that in this case, different from the polygonal landscapes we considered in Section \ref{sec:5:1}, the numerical experiments indicate that this model possesses the error threshold.

\begin{figure}[!th]
\centering
\includegraphics[width=0.45\textwidth]{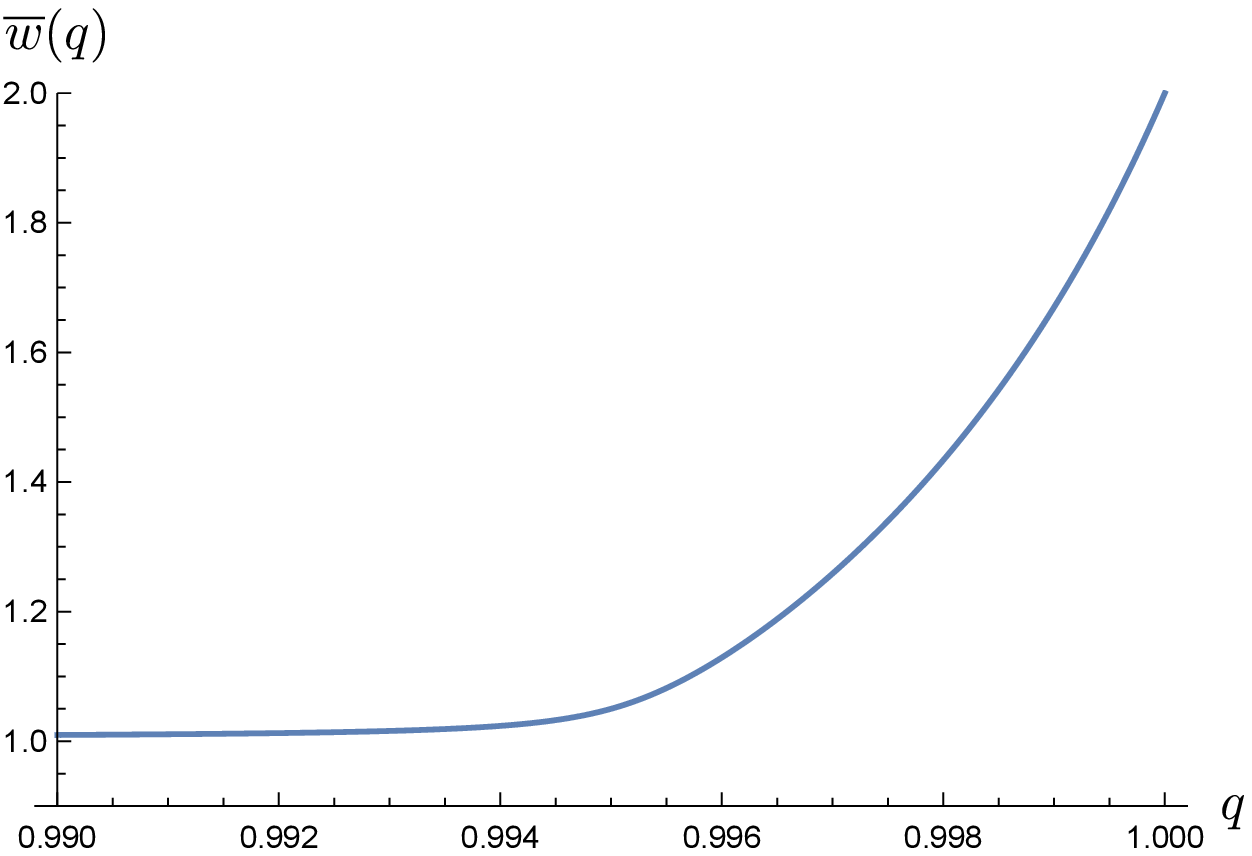}$\qquad$
\includegraphics[width=0.45\textwidth]{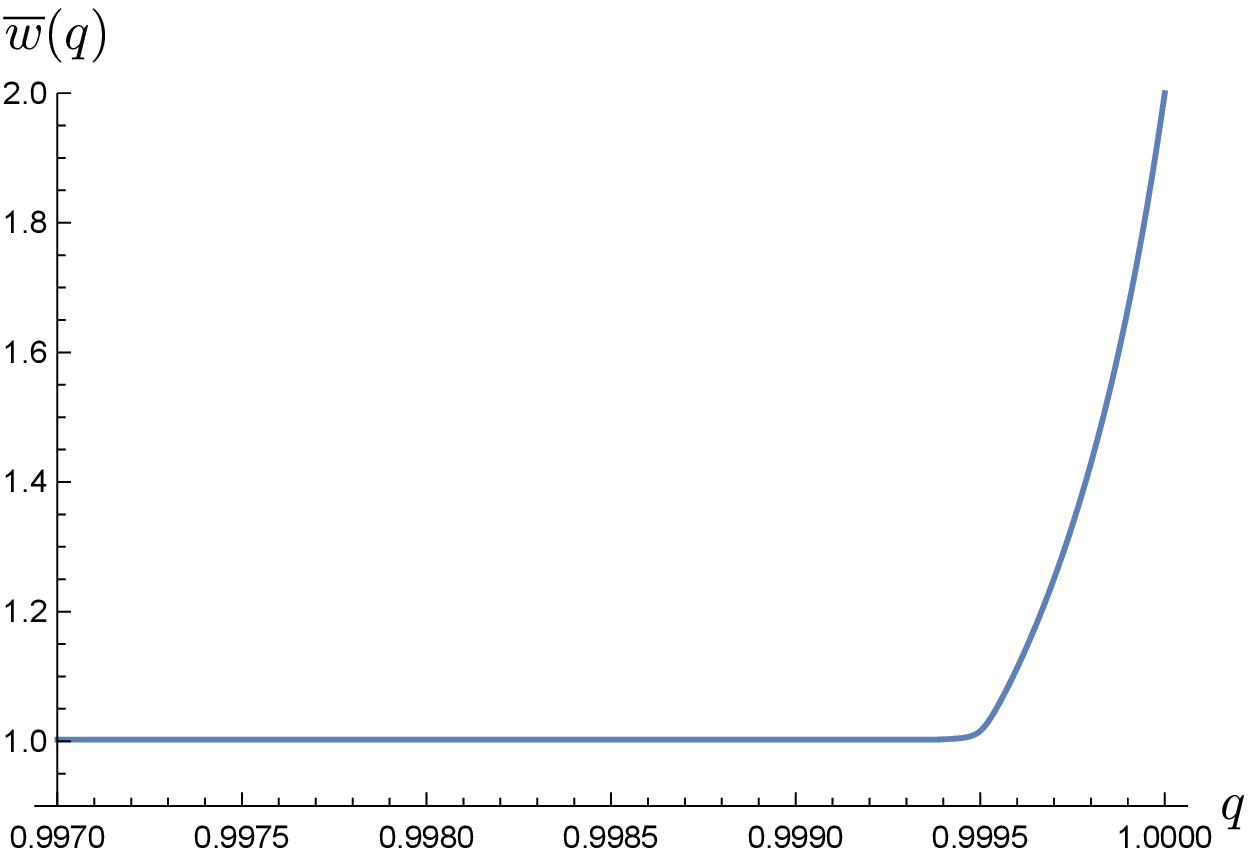}
\caption{The leading eigenvalue of spectral problem \eqref{1.2} on hyperoctahedral landscape for SPL $w=1,s=1$. On the left $n=100$, on the right $n=1000$.}\label{fig:9}
\end{figure}

\section{Concluding remarks}
The main contribution of the present paper is twofold. First, we introduced a generalized algebraic quasispecies model in which the standard binary hypercube of Eigen's model is replaced with an arbitrary finite metric space $X$. Second, we showed that if the structure of the fitness landscape is related to the isometry group of $X$ then a progress can be made in analytical investigation of the corresponding spectral problem. In particular, we found an explicit form of the algebraic equation for the leading eigenvalue (equation \eqref{2.23}).

At the same time, there are a number of open questions, which would be interesting to work on using the framework we suggest.

While the equation for $\overline{w}$ is written in the general form, in all the examples we considered here and in \cite{semenov2016eigen} we deal with the simplest case of two-valued fitness landscapes, when $X=A_0\bigsqcup A_1$. It is important to consider examples with more complicated partition of $X$. For example, the so-called mesa landscapes \cite{wolff2009robustness} have exactly this form.

The error threshold phenomenon (see Fig. \ref{fig:2}) was not analyzed in the present text. We remark that the error threshold was proven to exist for a simplicial mutation landscape in \cite{semenov2016eigen}. It looks plausible to conjecture that for the considered in the present text $m$-gon landscapes the error threshold is absent whereas for the hyperoctahedral mutation landscape it does exist. In general, we now have a more general question to ask: What are the properties of a finite metric space $X$ that guarantee the existence of the error threshold at least for some fitness landscapes $\bs w$?

Finally, more detailed analysis of the connections of the considered spectral problems with the Ising model is necessary. The proof that 2D Ising model possesses the phase transition, given by Onsager, is very non-elementary. On the other hand, for the simplicial and hyperoctahedral mutation landscapes the algebraic equation for the leading eigenvalue has degree 2 and 3 respectively and they provide much simpler examples of modeling systems that possess phase transition behavior.

\appendix

\section{Resulting tables}\label{ap:1}

In the following Table \ref{tab:1} several known homogeneous symmetric
triples $(X,d,\Gamma)$ of the landscapes are presented. Here $H_n$
is a hyperoctahedral group (the Weyl group of root system $B_n$ or
$C_n$ ) of order $2^n n!$, $S_n$ is a symmetric group, $D_n$ is a
dihedral group of order $2n$, $A_5$ is the alternating group of
order 60, $D_1\cong \Z/2\Z$, $\Gamma_0={\rm St}(x)$, $x\in X$. For
regular polytopes $P$ the metric space  $(X,d)$ consists of the
set $X=P^{(0)}$ of vertices, the metric $d$ is the edge metric
(see Example \ref{ex1.1}).
\begin{table}
$$\begin{array}{|l|c|c|c|c|c|c|c|}\hline
\;\mbox{Landscape}&X&d&l=|X|&\diam X&\Gamma&\Gamma_0
&r={\rm rk}_\Z R(X,d,\Gamma)\\
\hline\phantom{\rule{0pt}{12pt}} 1.\; \mbox{\it Hypercubic,}&\{0,1\}^n&Hamming&2^n&n&H_n&S_n&n+1\\
\mbox{\it or Eigen's}&&(edge)&&&&&\\&&metric&&&&&\\
\hline\phantom{\rule{0pt}{12pt}} 2.\; \mbox{\it Simplicial}&X_n&edge&n+1&1&S_{n+1}&S_n&2\\
&&metric&&&&&\\
\hline\phantom{\rule{0pt}{12pt}} 3.\; \mbox{\it Polygonal}&X_{2n}&edge&2n&n&D_{2n}&D_1&n+1\\
&&metric&&&&&\\
\hline\phantom{\rule{0pt}{12pt}} 4.\; \mbox{\it Polygonal}&X_{2n+1}&edge&2n+1&n&D_{2n+1}&D_1&n+1\\
&&metric&&&&&\\
\hline\phantom{\rule{0pt}{12pt}} 5.\, \mbox{\it Hyper-}&X_{n}&edge&2n&2&H_n&H_{n-1}&3\\
\mbox{\it octahedral}&&metric&&&&&\\
\hline\phantom{\rule{0pt}{12pt}} 6.\, \mbox{\it Dodeca-}&X=P^{(0)}&edge&20&5&A_5&D_3&6\\
\mbox{\it hedral}&&metric&&&&&\\\hline \phantom{\rule{0pt}{12pt}}
7.\, \mbox{\it Icosa-}&X=P^{(0)}&edge&12&3&A_5&D_5&4\\
\mbox{\it hedral}&&metric&&&&&\\\hline
\end{array}
$$
\caption{Examples of the homogeneous symmetric triples $(X,d,\Gamma)$}\label{tab:1}
\end{table}

In Table \ref{tab:2} the eigenpolynomials and their
multiplicities (in brackets) of the matrix ${\bs Q}$ are given.
The first one is always the (leading) distance polynomial of
multiplicity 1.

\begin{table}
$$\begin{array}{|l|c|c|}\hline
\;\mbox{Landscape}&X&\mbox{Eigenpolynomials of}\; {\bs Q}\; \mbox{(their multiplicities)} \\
\hline\phantom{\rule{0pt}{12pt}} 1.\; \mbox{\it Hypercubic,}&X_n&P_X(q)\equiv 1\;,\;(1),\\
\mbox{\it or Eigen's}&&P_j(q)=(2q-1)^j\;,\;\left(\binom{n}{j}\right),\quad j=1,\dots,n\\
\hline\phantom{\rule{0pt}{12pt}} 2.\; \mbox{\it Simplicial}&X_n&P_X(q)=q+n(1-q)\;,\;(1),\\
&&P_1(q)=2q-1\;,\;(n)\\
\hline\phantom{\rule{0pt}{12pt}} 3.\; \mbox{\it Polygonal}&X_{n},
&P_X(q)= q^N+2\sum\limits_{k=1}^{N-1}(1-q)^kq^{N-k}+(1-q)^N\;,\;(1),\\
&n=2N&P_j(q)=q^N+\sum\limits_{k=1}^{N-1} 2\cos(2\pi
kj/l)\,(1-q)^kq^{N-k}+(-1)^j(1-q)^N\;,\;(2),\\
&&j=1,\dots,N-1,\\&& P_N(q)=
q^N+2\sum\limits_{k=1}^{N-1}(-1)^k(1-q)^kq^{N-k}+(-1)^N(1-q)^N\;,\;(1),\\
\hline\phantom{\rule{0pt}{12pt}} 4.\; \mbox{\it Polygonal}&X_{n},
&P_X(q)= q^N+2\sum\limits_{k=1}^{N}(1-q)^kq^{N-k}\;,\;(1),\\
&n=2N+1&P_j(q)=q^N+\sum\limits_{k=1}^{N} 2\cos(2\pi
kj/l)\,(1-q)^kq^{N-k}\;,\;(2),\\
&&j=1,\dots,N\\
\hline\phantom{\rule{0pt}{12pt}} 5.\; \mbox{\it Hyper-}&X_{n}
&P_X(q)=q^2+(2n-2)(1-q)q+(1-q)^2\;,\;(1),\\
\mbox{\it octahedral}&&P_1(q)=(2q-1)^2\;,\;(n-1),\\
&&P_2(q)=2q-1\;,\;(n)\\
\hline\phantom{\rule{0pt}{12pt}} 6.\; \mbox{\it Dodeca-}&X
&P_X(q)=2q^4-4q^3+4q^2-2q+1\;,\;(1),\\
\;\mbox{\it hedral}&&P_1(q)=(-2q^4+4q^3+q^2-3q+1)(2q-1)\;,\;(4),\\
&&P_2(q)=(3q^4-6q^3+6q^2-3q+1+\sqrt{5}(q^4-2q^3+2q^2-q))\times\\
&&\times(2q-1)\;,\;(3),\\
&&P_3(q)=(3q^4-6q^3+6q^2-3q+1-\sqrt{5}(q^4-2q^3+2q^2-q))\times\\
&&\times(2q-1)\;,\;(3),\\
&&P_4(q)=(3q^2-3q+1)(2q-1)^2\;\;(4),\\
&&P_5(q)=(2q-1)^2\;,\;(5)\\
\hline\phantom{\rule{0pt}{12pt}} 7.\; \mbox{\it Icosa-}&X
&P_X(q)=-2q^2+2q+1\;,\;(1),\\
\;\mbox{\it hedral}&&P_1(q)=(q^2-q+1+\sqrt{5}(q^2-q))(2q-1)\;,\;(3),\\
&&P_2(q)=(q^2-q+1-\sqrt{5}(q^2-q))(2q-1)\;,\;(3),\\
&&P_3(q)=(2q-1)^2\;,\;(5)\\\hline
\end{array}
$$
\caption{Examples of eigenpolynomials and their multiplicities for several generalized mutation matrices $\bs Q$ corresponding to particular triples $(X,d,\Gamma)$}\label{tab:2}
\end{table}

\newpage


%
%
%
%
%
%
%
%
%
%
%
\end{document}